\documentclass[onecolumn]{article}
\usepackage{graphicx}
\usepackage{hyperref}
\usepackage{xcolor}
\usepackage{float}
\usepackage{amsmath,amsfonts,amssymb,amsthm,epsfig,epstopdf,titling,url,array}
\theoremstyle{plain}
\newtheorem{theorem}{Theorem}[section]
\newtheorem{lemma}[theorem]{Lemma}

\usepackage{algorithm}
\newtheorem{algori}{Procedure}[section]

\theoremstyle{remark}
\newtheorem{remark}{Remark}

\theoremstyle{definition}

\newtheorem{fac}[theorem]{Fact}

\begin{document}


\title{Mapping prior information onto LMI eigenvalue-regions for discrete-time subspace identification}

\author{Rodrigo A. Ricco, Bruno O. S. Teixeira}

\maketitle

\begin{abstract}
In subspace identification, prior information can be used to constrain the eigenvalues of the estimated state-space model by defining corresponding LMI regions. 
In this paper,  first we argue on what kind of practical information can be extracted from historical data or step-response experiments to possibly improve the dynamical properties of the corresponding model and, also, on how to mitigate the effect of the uncertainty on such information. For instance, prior knowledge regarding the overshoot, the period between damped oscillations and settling time may be useful to constraint the possible locations of the eigenvalues of the discrete-time model.
Then, we show how to map the prior information onto LMI regions and, when the obtaining regions are non-convex, to obtain convex approximations.
\end{abstract}

\section{Introduction}

Prior information can be used in system identification to possibly improve  some properties of the mathematical model \cite{Ljung2010,Barbosa2011,Noel2017}. 
Such gray-box approach is especially of interest when the dynamical data are limited in terms of persistence of excitation, signal-to-noise ratio, number of data samples, and, for nonlinear systems, coverage of operating points.
These situations may occur due to experiment restrictions on the plant or when only historical input-output data are available. 
By contrast, if {prior information} could be mathematically ``translated'' and properly incorporated in the estimation procedure, we may observe an improved performance compared to the black-box model. In fact, \cite{Teixeira2011} argue that if the uncertainty on the prior information is properly addressed, then it may improve the model accuracy.

In subspace identification, one of the main challenges is to insert prior information into the estimation procedure \cite{mercere2014,Markovsky2017}. The fact that the state-space matrices are estimated up to an unknown similarity transformation imposes an additional challenge \cite{Katayama2005c,verhaegen2007}. 
Recent works have addressed this topic. 
\cite{Lacy2003} shows how to guarantee the stability of the state-space model. The recursive case is addressed in \cite{Shang2015}.
\cite{privara2012} shows how to use the stationary gain as prior information in batch subspace identification methods. 
Likewise, \cite{Alenany2011} presents an alternative method to use prior information on the time constant and stationary gain into first-order models. In \cite{Alenany2013}, the time-varying case is addressed, where auxiliary information on both the stationary gain and some null entries in the transfer matrix associated to the multivariable system are taken into account.  
 \cite{mercere2016} derives the mathematical relations between  
 the stationary gain, damping ratio, time constant and natural frequency (\textit{prior} information) and the Markov parameters of the corresponding model. Then, such prior information is transformed into equality or inequality constraints. However, an algorithm to solve the discrete-time subspace identification subject to such constraints is not stated. \cite{Miller2013} presents a framework for the batch case, in which eigenvalue constraints are enforced by means of LMIs in the discrete-time subspace identification framework. Finally, in \cite{Demourant2017}, the constrained LMI-based frequency-domain subspace algorithm is applied to a wind tunnel test.   

In  control theory, the design of LMI regions is commonly based on the performance criteria previously defined by the user \cite{botto2002,Rosinova2014}. Conversely, for constrained identification purposes, we need to previously know some system properties  either from the physical laws that describe the system or from experimental data in order to define LMI regions. From the step-response tests, the auxiliary information regarding overshoot, the period between damped oscillations and the settling time seems to be reasonable way to define the LMI regions in practice. However, the auxiliary information obtained from experimental data may be uncertain due to many reasons, e.g., the measurement noise on data and the complexity of the system dynamics. In fact, one of the most challenging assumptions in the methods of \cite{Miller2013} and \cite{Demourant2017} is to consider that the prior information is already known on the $z$-plane. Another drawback is related to the convexity of the mapped complex regions obtained by the aforementioned dynamical features. 
Thus, it is of interest to approximate these mapped regions by means of convex regions \cite{botto2002,Rosinova2014}.  

In this scenario, the following question arises: how can we properly and approximately map the prior information from step response tests or historical data by means of eigenvalue constraints written as LMIs?  
We aim at circumventing the gap between mapping and using the prior information obtained in practice for discrete-time subspace identification with eigenvalue constraints. To achieve that, 
we assume that the dominant dynamics can be approximated by second order.

The connections between the practice and theory addressed in this paper allow for translating information regarding the overshoot, the period between damped oscillations and the settling time directly into LMI regions for discrete-time systems. In this issue, although the estimated values of the auxiliary information from step-response tests or even historical data are straightforward to obtain, we argue that tuning more conservative regions may overcome the problem on the approximation of the prior information.

Thus, the contribution of this work is twofold: (\rm{i}) a methodology to build LMI regions that constrain the model eigenvalues, according to the dominant system domains, from experimental noisy data is presented in Section \ref{sec:convex_approx}, and; (\rm{ii}) specifically, a novel more conservative approximation of the cardioid related to the overshoot in the $z$-plane is presented in Fact \ref{prop:ellipse_ricco}, in which the non-convex cardioid is mapped as an outer ellipse LMI region. 

This paper is organized as follows. Section \ref{sec:constraining_poles} states the problem under investigation, while Section \ref{background} presents important definitions. Section \ref{sec:convex_approx} presents a framework to build LMI regions to constrain the model eigenvalues. The numerical examples of Section \ref{results} illustrate the effectiveness of the proposed approaches. Finally, the concluding remarks are discussed in Section \ref{Conclusions}. 

\section{Problem statement}\label{sec:constraining_poles}

Consider the linear time-invariant discrete-time system
\begin{eqnarray}
x_{k+1}&=Ax_{k}+Bu_{k}+w_{k}, \nonumber \\ 
y_{k}&=Cx_{k}+Du_{k}+\nu_{k}, \label{state_eq}
\end{eqnarray}
where $A$ $\in$ $\mathbb{R}^{n\times n}$, $B$ $\in$ $\mathbb{R}^{n\times n_{u}}$, $C$ $\in$ $\mathbb{R}^{n_{y}\times n}$ and $D$ $\in$ $\mathbb{R}^{n_{y}\times n_{u}}$. 
The vectors $x_{k} \in \mathbb{R}^n$, $u_{k} \in \mathbb{R}^{n_{u}}$ and $y_{k} \in \mathbb{R}^{n_{y}}$ represent, respectively, the states, inputs and outputs. $w_{k} \in \mathbb{R}^n$ and $\nu_{k} \in \mathbb{R}^{n_{y}}$ are the process and measurement noise terms, respectively, both assumed to be zero-mean white Gaussian noise sequences.

Assume that an initial state-space model given by \linebreak $\{A^{*},B^{*},C^{*},D^{*}\}$ is obtained by means of a standard subspace identification method \cite{Katayama2005c,verhaegen2007}. Also, assume that  prior information about the eigenvalues of $A$ is available. 
Given these assumptions, \cite{Miller2013} show how to enforce constraints (prior information) on the localization of the eigenvalues of the matrix $A$ of the model \eqref{state_eq} by modifying the initial estimate $A^{*}$ as follows. 

Consider the  cost function
\begin{eqnarray}\label{cost_P_Q}
J_{X}(Q,P)\triangleq\left\|A^{*}P~~-~~Q\right\|^{2}_{F},
\end{eqnarray}
where $P \in \mathbb{R}^{n \times n}$ is assumed to be a symmetric matrix, and  define $Q$ $\triangleq$ $\hat{A}P$ in order to obtain a linear optimization problem, where $\hat{A}$ is the new estimate to be obtained. 
The problem of {\em subspace identification with eigenvalue constraints} is given by
\begin{eqnarray} 
\text{minimize}&~~J_{X}(Q,~P),  \label{eq:eigenvalue_c_problem} \\
\text{subject~to}&~~M_{\mathcal{D}}(Q,~P)\geq 0, \label{eq:D_constraint} \\ 
\label{eq:P_matrix} &P=P^{T}>0, 
\end{eqnarray}
where $J_{X}(Q,P)$ is given by \eqref{cost_P_Q} and
$M_{\mathcal{D}}(Q,P)~\triangleq~\lambda \otimes P~+~\beta \otimes Q+\beta^{T}  \otimes Q^{T}$
is the eigenvalue constraint written as a LMI corresponding to the 
 the convex set $\mathcal{D}$ of the $z$-plane defined by
\begin{eqnarray}\label{Destabilidade}
\mathcal{D}\triangleq\left\{z~\in~\mathbb{C}: f_{\mathcal{D}}(z) \geq 0 \right\},
\end{eqnarray}
where
$f_{\mathcal{D}}(z)\triangleq\lambda~+~\beta z~+~\beta^{T}\bar{z}$
is the characteristic function of $\mathcal{D}$, where $\lambda$ is a symmetric matrix and $\beta$ is a square matrix. 
If the problem given by \eqref{eq:eigenvalue_c_problem}-\eqref{eq:P_matrix}  is feasible, then we obtain the matrices $Q$ and $P$, from which we estimate
$\hat{A}=QP^{-1}$
with eigenvalues belonging to the convex set \eqref{Destabilidade}. 
From \eqref{eq:eigenvalue_c_problem}-\eqref{eq:P_matrix}, note that it is not possible to enforce constraints on each eigenvalue of $\hat{A}$ separately. 

The goal of this paper is to obtain LMI constraints \eqref{eq:D_constraint} using the D-stability theory \cite{Chilali1996} 
to solve the aforementioned gray-box subspace identification problem.
To achieve that, we assume that uncertain prior information regarding (i) the overshoot $O_{s}$, (ii) the period between damped oscillations $T_d$, or (iii) the settling time $t_{s,1\%}$ may be available from step-response tests or historical data, for instance. Then, assuming that the dominant dynamics can be approximated by second order, we show how to map such prior information onto convex LMI eigenvalue constraints \eqref{eq:D_constraint} on the $z$-plane. Also, we discuss on how such prior information is useful to improve the state-space models. 
\section{Preliminaries}\label{background}

The dominant dynamic  behavior of many practical processes can be approximated by second order time-invariant models
\begin{eqnarray} 
\label{eq:segunda_ordem}
G(s)=\frac{Kw_{n}^{2}}{s^{2}+2\zeta w_{n}s + w_{n}^2},
\end{eqnarray}
where 
$\zeta$ is the damping ratio, $w_n$ is the natural frequency, and $K$ is the static gain. The poles of \eqref{eq:segunda_ordem} are given by the real  \hbox{$\sigma \triangleq \zeta w_{n}$} and   imaginary  \hbox{$w_{d} \triangleq w_{n}\sqrt{1-\zeta^2}$} parts. 
The transient response of \textit{underdamped} second-order systems is characterized by the dynamical measures:  the \textit{time-constant} $\tau=1/\zeta w_{n}$, the \textit{overshoot} \hbox{$O_{s} \approx 100(1-\zeta/0.6)$}, the \textit{settling time} $t_{s,1\%} \approx 4.6/\zeta w_{n}$, the \textit{rise-time} $t_{r} \approx 1.8/w_{n}$, the \textit{peak-time} $t_{p}=\pi/w_{d}$ and the \textit{period between damped oscillations} $T_d \triangleq 2\pi/w_{d}$. For details about the relations aforementioned see \cite{Franklin2010}.
Recall that although the overshoot $O_{s}$ is defined in the sense of control systems, here it is used as a measure of the maximum oscillation of the \textit{underdamped} system.
%


In order to represent the dynamical regions for \textit{underdamped} systems on the $s$-plane, assume that $O_{s}\leq O^{\textrm{max}}_{s}$,   
$T_d \geq T_d^{\textrm{max}}$ and $t_{s,1\%} \leq t^{\textrm{max}}_{s,1\%}$. 
These assumptions are motivated by the following practical reasons: (i) we cannot estimate  these dynamical measures exactly; (ii) it is reasonable to be more conservative on the definition of the corresponding LMI regions; and (iii) in doing so, we  consider the effect of additive noise.
Then, rewriting these relations, we obtain 
\begin{eqnarray} \label{eq:Mp_max}
O_{s} \leq O^{\textrm{max}}_{s}   &\Longrightarrow & \zeta  \geq 0.6\left(1-\frac{O^{\textrm{max}}_{s}}{100}\right),\\
\label{eq:T_max}
T_d  \geq T_d^{\textrm{max}} & \Longrightarrow & w_{d}  \leq \frac{2 \pi}{T_d^{\textrm{max}}},\\ \label{eq:ts_max}
t_{s,1\%} \leq t^{\textrm{max}}_{s,1\%} & \Longrightarrow  & \zeta w_{n}   \geq \frac{4.6}{t^{\textrm{max}}_{s,1\%}},
\end{eqnarray}
where $\zeta \geq \zeta^{\rm{min}} $, $w_{d} \leq w_{d}^{\textrm{max}}$, and $\zeta w_{n} \geq {\zeta w_{n}}^{\textrm{min}}$, such that \hbox{$0<\zeta^{\rm{min}} <1$}, $w^{\rm{max}}_{d}>0$ and ${\zeta w_{n}} ^{\rm{min}}>0$.
Fig. \ref{fig:mapping_regions1}a shows that the region described by \eqref{eq:Mp_max} is  bounded on the $s$-plane inside the cone defined by two lines with angle \hbox{$\pm\beta=\textrm{cos}^{-1}(\zeta)$}. 
The region (\ref{eq:T_max}) is given by  lines positioned parallel to the real axis in $\pm w_{d}$ as shown in Fig \ref{fig:mapping_regions2}a. Finally, Fig. \ref{fig:mapping_regions3}a illustrates the region  \eqref{eq:ts_max} as the semiplane on the left of the  line $\sigma=-\zeta w_{n}$. Observe that the meaning of  \eqref{eq:Mp_max}-\eqref{eq:ts_max} can be analyzed by means of the figures \ref{fig:mapping_regions1}a, \ref{fig:mapping_regions2}a and \ref{fig:mapping_regions3}a on the $s$-plane and also by the figures \ref{fig:mapping_regions1}b, \ref{fig:mapping_regions2}b and \ref{fig:mapping_regions3}b on the $z$-plane. For example, although $\zeta^{2}>\zeta^{1}$ ($O_{s,2} < O_{s,1}$), note that the region of $\zeta^{1}$ is larger than the region of $\zeta^{2}$. 


We know that the poles of the continuous-time model $s_{1,2}=-\zeta w_{n} \pm \jmath w_{n}\sqrt{1-\zeta^2}$ are mapped onto $z_{1,2}=e^{s_{1,2}T_{s}}=z_{1,2}=re^{\jmath \theta}$ with 
$r \triangleq e^{-\zeta w_{n}T_{s}}$ and
$\theta \triangleq  \pm w_{n}T_{s}\sqrt{1-\zeta^2}$,
where $T_{s}$ is the sampling period. 
The regions on the left of $\zeta$ on the $s$-plane (Fig. \ref{fig:mapping_regions1}a) are mapped within the cardioids on the $z$-plane (Fig. \ref{fig:mapping_regions1}b).
The parallel lines in Fig. \ref{fig:mapping_regions2}a 
are mapped in Fig. \ref{fig:mapping_regions2}b on the right of $w_{d}$ on the $z$-plane. Note that the bottom-half plane mapping from the $s$-plane into the $z$-plane could be analyzed by symmetry.   
Finally, figures \ref{fig:mapping_regions3}a and b shows two illustrative cases of regions mapped from the $s$-plane onto the $z$-plane regarding $\zeta w_{w}$.

\section{LMI eigenvalue dynamical regions for subspace identification}\label{sec:convex_approx}
The regions presented in figures \ref{fig:mapping_regions1}b, \ref{fig:mapping_regions2}b and \ref{fig:mapping_regions3}b may be approximated or exactly represented by convex LMI functions on the $z$-plane. 
To combine different eigenvalue regions \eqref{Destabilidade}, the next result obtained in \cite{Chilali1996} is of interest.
 \begin{lemma}[\cite{Chilali1996}] \label{lem:intersection} {Given $N$ LMI regions $\{\mathcal{D}_{1}$, $\cdots$, $\mathcal{D}_{N}\}$, the intersection of these regions {$\mathcal{D}$ $=$ $\mathcal{D}_{1}\cap \cdots \cap \mathcal{D}_{N}$} has the following characteristic function
\begin{eqnarray}\label{LMIs_juntas}
f_{\mathcal{D}}(z)=\left[\begin{array}{ccc}
                f_{\mathcal{D}_{1}}(z) &\cdots &0 \\
                \vdots  &\ddots &\vdots\\
                 0   &\cdots  &f_{\mathcal{D}_{N}}(z)
              \end{array}
\right].
\end{eqnarray}
}\end{lemma}

In order to parametrize the regions shown in figures \ref{fig:mapping_regions1}b, \ref{fig:mapping_regions2}b, and \ref{fig:mapping_regions3}b, the following sections discuss the theoretical aspects of how these regions are mapped into LMIs. In sections \ref{sec:Overshoot_revisited},  \ref{sec:period_oscillations} and \ref{sec:settling_time} we focus on the discussion of limitations and benefits of  using such LMIs for subspace identification. Also, practical aspects in the use of auxiliary information is discussed in Section \ref{Sec:Bound_variables} and a new LMI region regarding $O_{s}$ (see Fig. \ref{fig:mapping_regions1}b) is proposed in Fact \ref{prop:ellipse_ricco}.

  \begin{figure}[H]
\begin{center}
\includegraphics[scale=0.5]{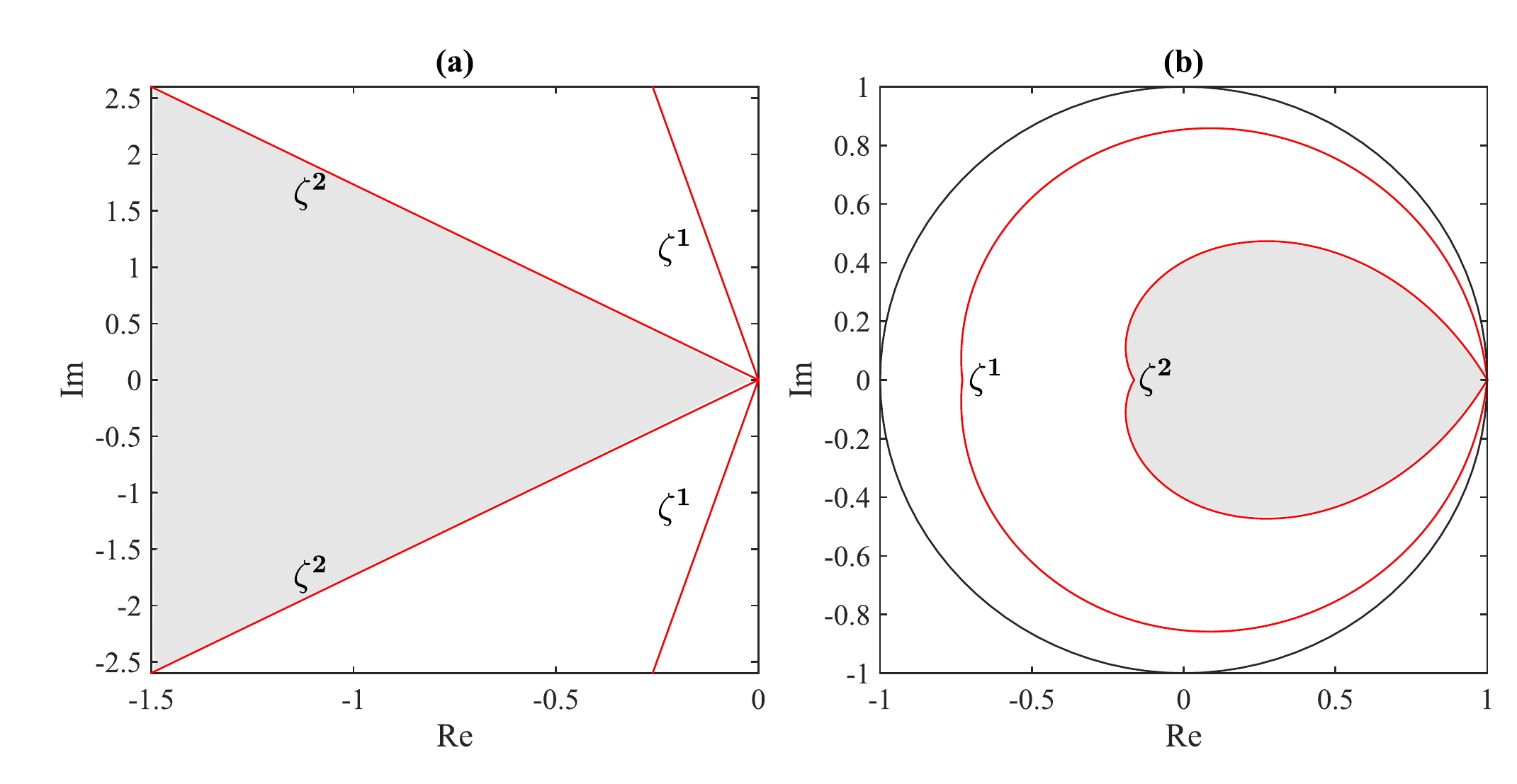} 
\end{center}
\caption{Some dynamical regions mapped from the $s$-plane onto the $z$-plane. (a) the region inside the cones bounded by lines according to $\zeta$ are mapped within the cardioids presented in (b). The shaded regions on the left (continuous) are mapped on shaded regions on the right (discrete-time).}
\label{fig:mapping_regions1}                                 
\end{figure}  
  \begin{figure}[H]
\begin{center}
\includegraphics[scale=0.5]{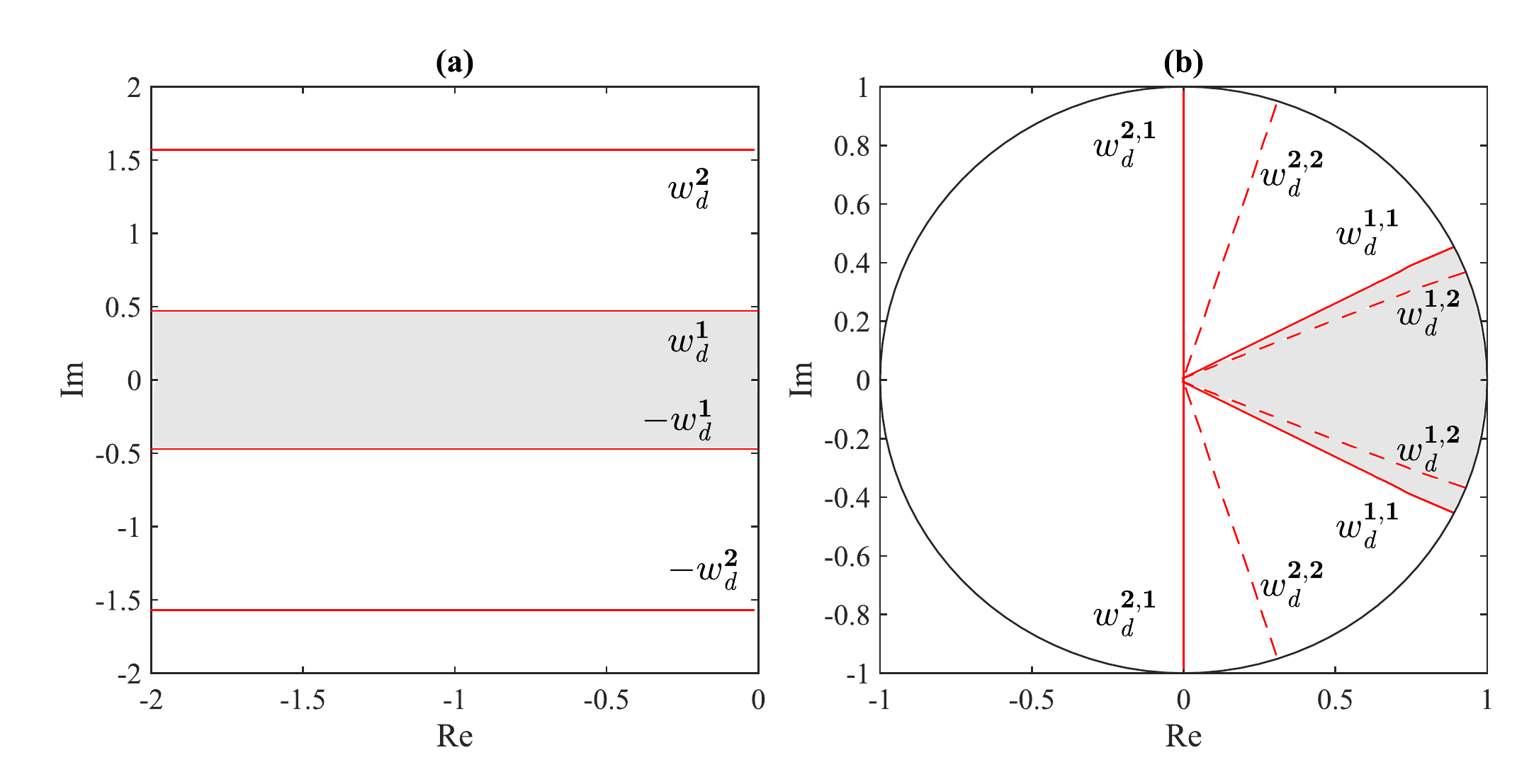} 
\end{center}
\caption{Illustrative dynamical regions mapped from the $s$-plane onto the $z$-plane. The region between the horizontal lines $\pm w_{d}$ is given by (a) mapped inside the region limited by the cones generated from $w_{d}$ as  in (b) for two possible values of the sampling period, where $T^{1}_{s} < T^{2}_{s}$ in red and dashed red, respectively. The shaded regions on the left (continuous) are mapped on shaded regions on the right (discrete-time). }
\label{fig:mapping_regions2}                                 
\end{figure}
  \begin{figure}[H]
\begin{center}
\includegraphics[scale=0.5]{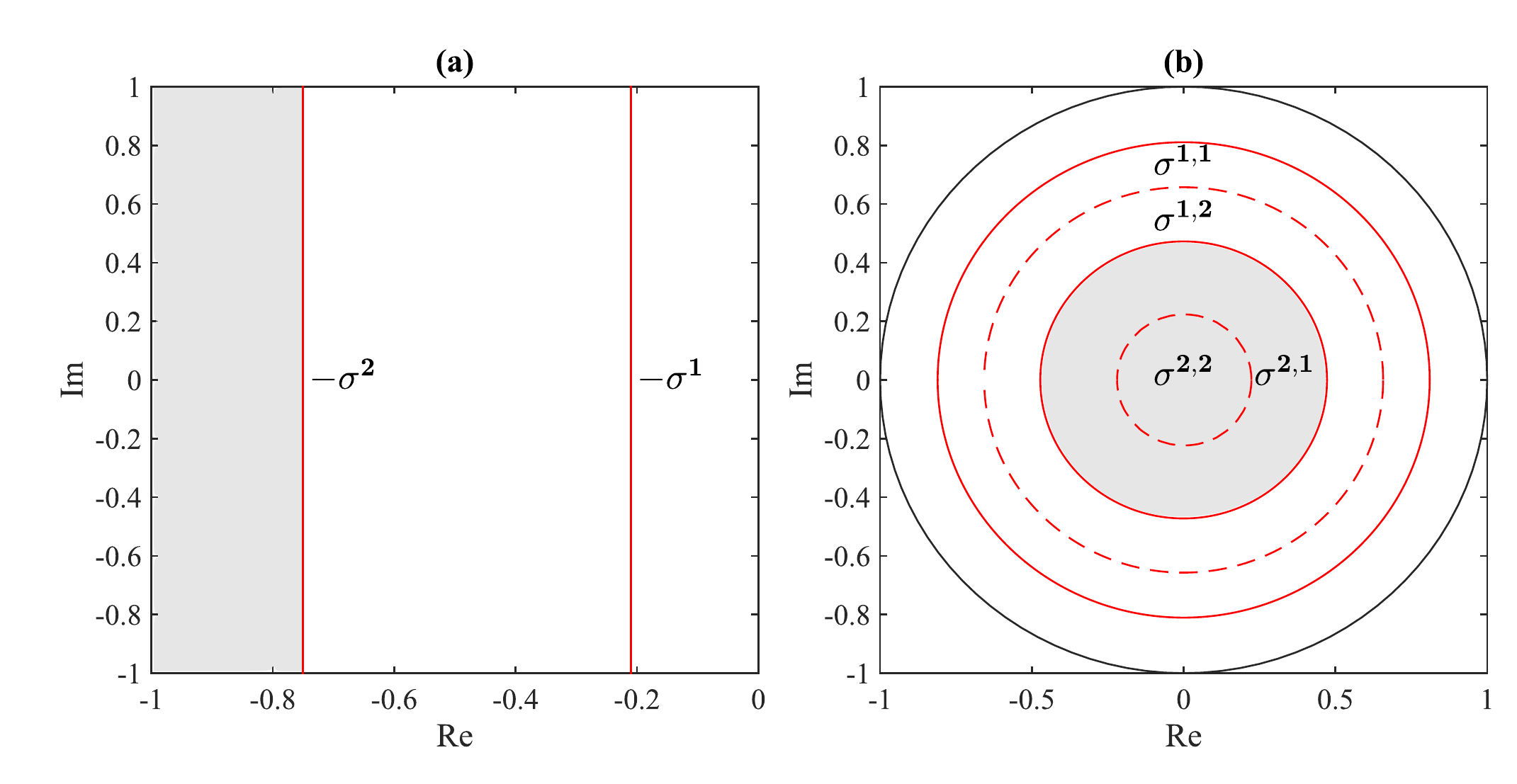} 
\end{center}
\caption{Examples of dynamical regions mapped from the $s$-plane onto the $z$-plane. The region left to the vertical lines bounded by $\zeta w_{n}$ given by (a) is mapped within the circles shown in (b) for two illustrative values of the sampling period, where $T^{1}_{s} <T^{2}_{s}$ in red and dashed red, respectively. The shaded regions on the left (continuous) are mapped on shaded regions on the right (discrete-time).}
\label{fig:mapping_regions3}                                 
\end{figure}

\subsection{Overshoot} \label{sec:Overshoot_revisited}
Next, we review the results of \cite{Rosinova2014} in which the region corresponding to $O_{s}$ (see Fig. \ref{fig:mapping_regions1}b) is approximated either by inner circles or ellipses on the $z$-plane.
  
 \begin{fac}[\cite{Rosinova2014}]\label{fact:circle}The set  approximating the cardioid related to $\zeta^{\rm{min}} $ by means of a circle is given by
\begin{eqnarray}\label{def_abs_z}
\mathcal{D}_{O}=\left\{z~\in~\mathbb{C}: |z|\leq r^{\rm{max}} ,~0~<~r^{\rm{max}} <~1 \right\}
\end{eqnarray}
and is equivalent to the LMI  $f_{\mathcal{D}_{O}}(z) \geq 0$, where
\begin{eqnarray}\label{eq:mp_1}
f_{\mathcal{D}_{O}}(z)=\left[\begin{array}{cc}
                r^{\rm{max}} & -c^{\rm{max}} \\
                -c^{\rm{max}} & r^{\rm{max}}
              \end{array}
\right]+\left[\begin{array}{cc}
                0 & 1 \\
                0 & 0
              \end{array}
\right]z+\left[\begin{array}{cc}
                0 & 0 \\
                1 & 0
              \end{array}
\right]\bar{z},
\end{eqnarray}
 
\noindent such that 
  
\begin{eqnarray}\label{eq:cmax}
c^{\rm{max}}&\triangleq& e^{-\beta^{\rm{max}}/\rm{tan}(\beta^{\rm{max}})}\rm{cos}(-\beta^{\rm{max}}),\\ \label{eq:rmax}
r^{\rm{max}}&\triangleq& e^{-\beta^{\rm{max}}/\rm{tan}(\beta^{\rm{max}})}\rm{sin}(\beta^{\rm{max}}),
\end{eqnarray}
where $\beta^{\rm{max}}\triangleq \rm{cos}^{-1}(\zeta^{\rm{min}} )$, $c^{\rm{max}}$ is the center and $r^{\rm{max}}$ is the radius of the circle. 
\end{fac}
  
 \begin{fac}[\cite{Rosinova2014}]\label{fact:ellipse_rosinova}The set  approximating the cardioid related to $\zeta^{\rm{min}} $ by means of an ellipse is given by
\begin{eqnarray}\label{def_abs_z}
\mathcal{D}_{O}=\left\{z~\in~\mathbb{C}: |z|\leq r_{e}^{\rm{max}} ,~0~<~r_{e}^{\rm{max}} <~1 \right\}
\end{eqnarray}
and is equivalent to the LMI  $f_{\mathcal{D}_{O}}(z) \geq 0$, where
\begin{eqnarray}\label{eq:mp_2}\nonumber
f_{\mathcal{D}_{O}}(z)&=&\left[\begin{array}{cc}
                1 & -e c^{\rm{max}} \\
                -e c^{\rm{max}} & 1
              \end{array}
\right]\\  \nonumber
&+&\left[\begin{array}{cc}
                0 & (e-f)/2 \\
                (e+f)/2 & 0
              \end{array}
\right]z\\
&+&\left[\begin{array}{cc}
                0 & (e+f)/2 \\
                (e-f)/2 & 0
              \end{array}
\right]\bar{z}, \end{eqnarray}
such that 
\begin{eqnarray}\label{eq:a_major}
a&\triangleq& c^{\rm{max}}+e^{-\pi/\rm{tan}(\beta^{\rm{max}})},\\ \label{eq:b_minor} 
b&\triangleq& r^{\rm{max}},\\  \label{eq:e_ellipse}
e&\triangleq&1/a,\\ \label{eq:f_ellipse}
f&\triangleq&1/b,
\end{eqnarray}
where $c^{\rm{max}}$ is the center, $a$ is the real semi-axis, $b$ is the imaginary semi-axis, and $r_{e}^{\rm{max}}$ is the distance from the center $c^{\rm{max}}$  and an arbitrary point $P$ of the ellipse and is given in polar coordinates by
\begin{eqnarray}\label{eq:distance_center_1}
r_{e}^{\rm{max}}\triangleq\frac{a b}{\sqrt{a^{2}\rm{sin}^{2}(\vartheta)+b^{2}\rm{cos}^2(\vartheta)}},
\end{eqnarray}
where $\vartheta$ is the angle between $r_{e}^{\rm{max}}$ and the real axis. 
 \end{fac}
 
\begin{remark}
It is assumed that $0 < \zeta^{\rm{min}} < 1$ in Facts \ref{fact:circle}-\ref{fact:ellipse_rosinova}. Replacing $\zeta^{\rm{min}}\rightarrow 0$ in \eqref{eq:cmax}-\eqref{eq:rmax}, we obtain $c^{\rm{max}}\rightarrow 0$ and $r^{\rm{max}}\rightarrow 1$. Given \eqref{eq:a_major}-\eqref{eq:b_minor}, if  $c^{\rm{max}}\rightarrow 0$  and $r^{\rm{max}}\rightarrow 1$, then $a\rightarrow 1$ and $b\rightarrow 1$. Conversely, if $\zeta^{\rm{min}}\rightarrow 1$ in \eqref{eq:cmax}-\eqref{eq:rmax}, then $c^{\rm{max}}\rightarrow e^{-1}$  and \hbox{$r^{\rm{max}}\rightarrow 0$}. Likewise, we have $a
\rightarrow \nexists$  and $b\rightarrow 0$. Thus, Facts \ref{fact:circle}-\ref{fact:ellipse_rosinova} have been proposed just for {\it underdamped} systems. We also note that the regions where $\zeta^{\rm{min}} \rightarrow 1$ tend to be too small and cannot be well approximated by the LMIs given in Facts \ref{fact:circle}-\ref{fact:ellipse_rosinova}.
 \end{remark}
 
Figures \ref{fig:mp_circles_ellipses}a and \ref{fig:mp_circles_ellipses}b show the inner approximation of cardioids by means of circles (Fact \ref{fact:circle}) and  ellipses (Fact \ref{fact:ellipse_rosinova}), respectively. Note that, the larger $\zeta^{\rm{min}}$ is, the better is the approximation provided by the ellipse compared to the circle.
 \begin{figure}[H]
\begin{center}
\includegraphics[scale=0.5]{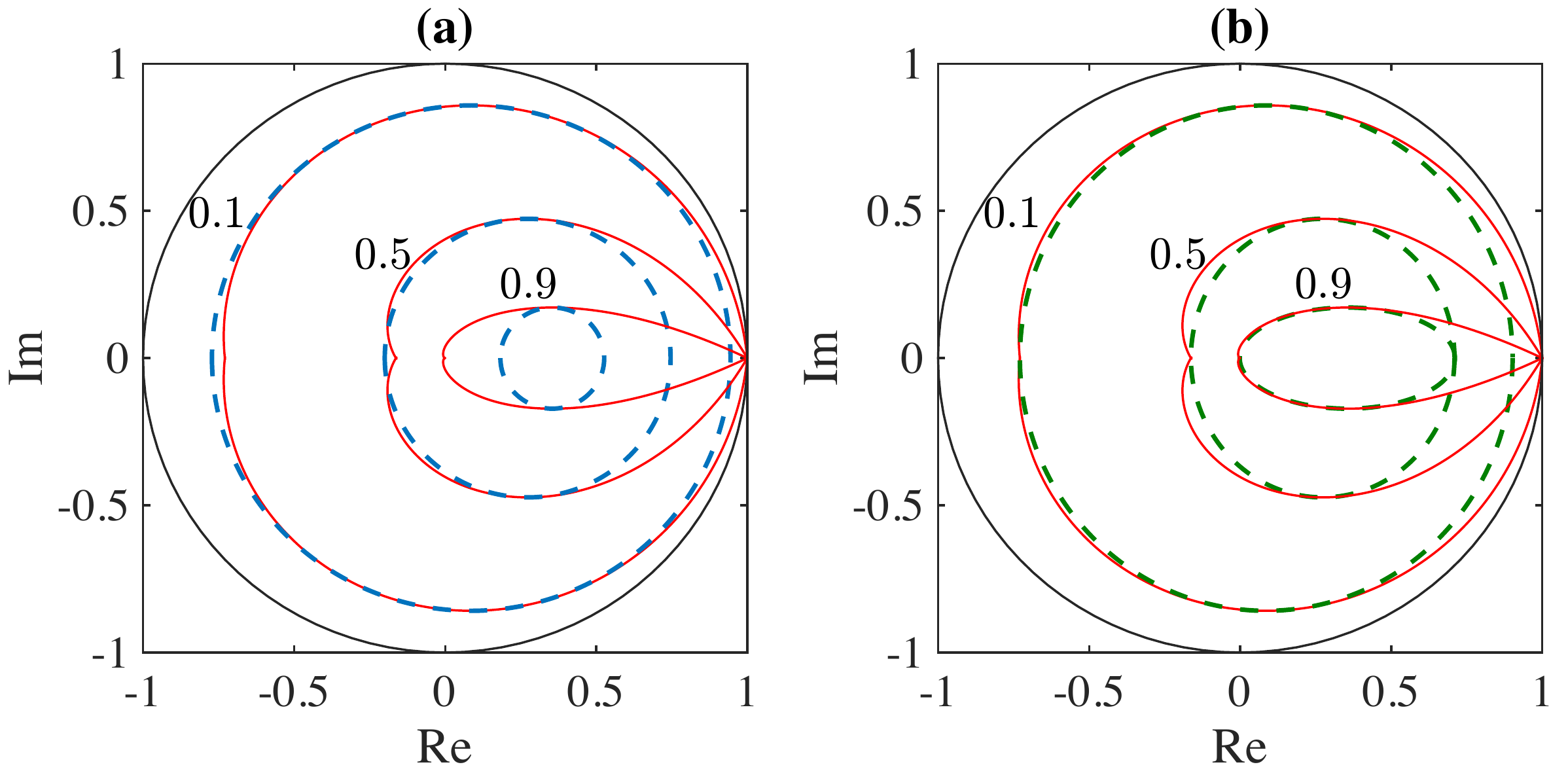} 
\end{center}
\caption{Convex approximations of the auxiliary information $O^{\textrm{max}}_{s}$ (cardioids in red for $\zeta^{\rm{min}} =\{0.1, 0.5, 0.9\}$ ): (a) using  circles (Fact \ref{fact:circle}) in dashed blue and (b) using ellipses (Fact \ref{fact:ellipse_rosinova}) in dashed green. 
The more damped the system is, the smaller is the area of the corresponding ellipse on the $z$-plane.}
\label{fig:mp_circles_ellipses}                                 
\end{figure}  

Next, we discuss under which conditions the convex ellipsoidal region described by (\ref{eq:mp_2}) is more (or not) useful than the circle (\ref{eq:mp_1}). Indeed, if we vary $\zeta^{\rm{min}} $, then the eccentricity of the ellipse (\ref{eq:mp_2}) also varies and generate three possible different regions. 
\begin{remark}\label{rem:ellipse_vertical} If $a=b$, then the ellipse  (\ref{eq:mp_2}) degenerates  into the circle  (\ref{eq:mp_1}). If $a< b$ then the region (\ref{eq:mp_2}) becomes a vertically-oriented ellipse in which the imaginary semi-axis $b$ is larger than real semi-axis $a$.  
In this case, we verify that the region (\ref{eq:mp_2}) is useless for  approximating the cardioid region for underdamped systems, such that the circle approximation (\ref{eq:mp_1}) should be used rather.
Conversely, if $a> b$ then the region (\ref{eq:mp_2}) becomes a horizontally-oriented ellipse which is useful to approximate the cardioid.
\end{remark}
Now we present the results that indicate for which values of $\zeta^{\rm{min}} $ the ellipse (\ref{eq:mp_2}) is horizontally-oriented like the cardioid it approximates.
 \begin{fac} \label{rem:ellipse_horizontal} 
Assume that $\zeta^{\rm{min}} \in (0,~1)$ for underdamped systems. If $0.6128<\zeta^{\rm{min}} <1$, then the region (\ref{eq:mp_2}) becomes a horizontally-oriented ellipse in which the real semi-axis $a$ is larger than imaginary semi-axis $b$.
 \end{fac}
\begin{proof}
From $a=b$ and the equations (\ref{eq:cmax})-(\ref{eq:rmax}), we have 
\begin{eqnarray}\label{eq:proof_circle_eq}
\gamma \rm{cos}(-\beta^{\rm{max}})+e^{-\pi/\rm{tan}(\beta^{\rm{max}})}=\gamma \rm{sin}(-\beta^{\rm{max}}),
\end{eqnarray}
where $\gamma=e^{-\beta^{\rm{max}}/\rm{tan}(\beta^{\rm{max}})}$. By using numerical solvers for (\ref{eq:proof_circle_eq}) and $\zeta^{\rm{min}} \in (0,~1)$, we obtain $\beta^{\rm{max}} \approx 52.2^{\circ}$ and $\zeta^{\rm{min}} \approx 0.6128$. Therefore, $\zeta^{\rm{min}} \approx 0.6128$ is the critical value for which the ellipse described by (\ref{eq:mp_2})  changes from the vertical orientation \hbox{($0<\zeta^{\rm{min}} <0.6128$}; see Remark \ref{rem:ellipse_vertical}) to the horizontal orientation ($0.6128 < \zeta^{\rm{min}} < 1$). 
\end{proof}
In system identification, we are not interested in the inner approximations \eqref{eq:mp_1}  and \eqref{eq:mp_2}  of the region described by $\zeta^{\rm{min}} $ on the $z$-plane. For this reason, the next result rewrites the parametrization of the real semi-axis $a$ and the center $c^{\textrm{max}}$ such that the new $D$-stability region could be more useful for constrained subspace identification. In so doing, the issues raised in Remark \ref{rem:ellipse_vertical} and  Fact $\ref{rem:ellipse_horizontal}$ are circumvented.

\begin{fac}\label{prop:ellipse_ricco}
Assume that the parameters $c^{\rm{max}}$, $a$ and $e$ are given by (\ref{eq:cmax}), (\ref{eq:a_major}) and (\ref{eq:e_ellipse}), respectively. The new ellipse approximating the cardioid corresponding to $\zeta^{\rm{min}}$ is given by \eqref{def_abs_z}-\eqref{eq:mp_2}, replacing $c^{\rm{max}}$ by $c_{n}^{\rm{max}}\triangleq c^{\rm{max}}\mu$, $a$ by $a_{n} \triangleq a \mu$, $e$ by $e_{n}\triangleq1/a_{n}$, where  ${\displaystyle \mu \triangleq \frac{1}{a+c^{\rm{max}}}}$ and $r_{e,n}^{\rm{max}}$ is the new distance from the new center $c_{n}^{\rm{max}}$ given in a similar way as presented in (\ref{eq:distance_center_1}).
\end{fac}
\begin{proof} Note that we propose here a new parametrization of the ellipse defined in Fact \ref{fact:ellipse_rosinova}. For more details, the reader if referred to corresponding proof \cite[Section~3]{Rosinova2014}. The polar transformation of the real semi-axis $x$ is given by
\begin{eqnarray}\label{eq:x_polar}
x\triangleq a\textrm{cos}(\vartheta)+c^{\rm{max}}. 
\end{eqnarray} 
The maximum value of $x$ is achieved  when $\vartheta=0$ or $\vartheta=2\pi$ where $\vartheta \in \left[0,~2\pi \right]$. Hence, the latter equation can be written as
\begin{eqnarray}\label{eq:x_polar_max}
x^{\rm{max}}=a+c^{\rm{max}}. 
\end{eqnarray} 
Consider that the maximum value of $x^{\rm{max}}$ should be equal to $1$ (limited by the unit circle) and also that the ellipse becomes a circle if $x^{\rm{max}}=1$. In other words, we guarantee that the new re-parameterized ellipse is horizontally-oriented for $0<\zeta^{\rm min}<1$. To guarantee that $x$ is limited by the unit circle ($x^{\rm{max}}=1$), multiply both sides of (\ref{eq:x_polar}) by the inverse of the right part of (\ref{eq:x_polar_max}), such that 
\begin{eqnarray} \nonumber
x_{n}=\frac{x}{a+c^{\rm{max}}}&=&\frac{a}{a+c^{\rm{max}}}\textrm{cos}(\vartheta)+\frac{c^{\rm{max}}}{a+c^{\rm{max}}},\\ 
&=&a \mu \textrm{cos}(\vartheta)+ c^{\rm{max}}\mu,  \label{eq:x_polar_n}
\end{eqnarray}
where $\mu \triangleq \frac{1}{a+c^{\rm{max}}}$.
Comparing (\ref{eq:x_polar_n}) to (\ref{eq:x_polar}), we verify that the new horizontal ellipse, limited by the unity circle, have the following new parameters $a_{n}\triangleq a \mu$
and $c_{n}^{\rm{max}}\triangleq c^{\rm{max}}\mu$, 
where $a_{n}$ is the new radius and $c_{n}^{\rm{max}}$ is the new center of the ellipse region (\ref{eq:mp_2}). To complete the proof we also note that the maximum value of (\ref{eq:x_polar_n}) is equal to $1$  ($x_{n}^{\rm{max}}=1$) when $\vartheta=0$ or $\vartheta=2\pi$.
\end{proof}
Fig. \ref{fig:mp_region_z_ellipse3} compares the convex approximations of the cardioid using the characteristic equation of the more conservative ellipse defined in Fact \ref{prop:ellipse_ricco}. Note that, unlike \eqref{eq:mp_1} or \eqref{eq:mp_2} in Fig. \ref{fig:mp_circles_ellipses}a-b, the proposed ellipse defined by Fact \ref{prop:ellipse_ricco} encloses most of the cardioid area.
It is important to point out that, unlike the other ellipse approximations, only the ellipse re-parameterized  by Fact \ref{prop:ellipse_ricco} encompasses the region with eigenvalues close to $1$. This region is very important for system identification due to the effect of high sampling in the localization of the system poles, for instance.
 \begin{figure}[H]
\begin{center}
\includegraphics[scale=0.5]{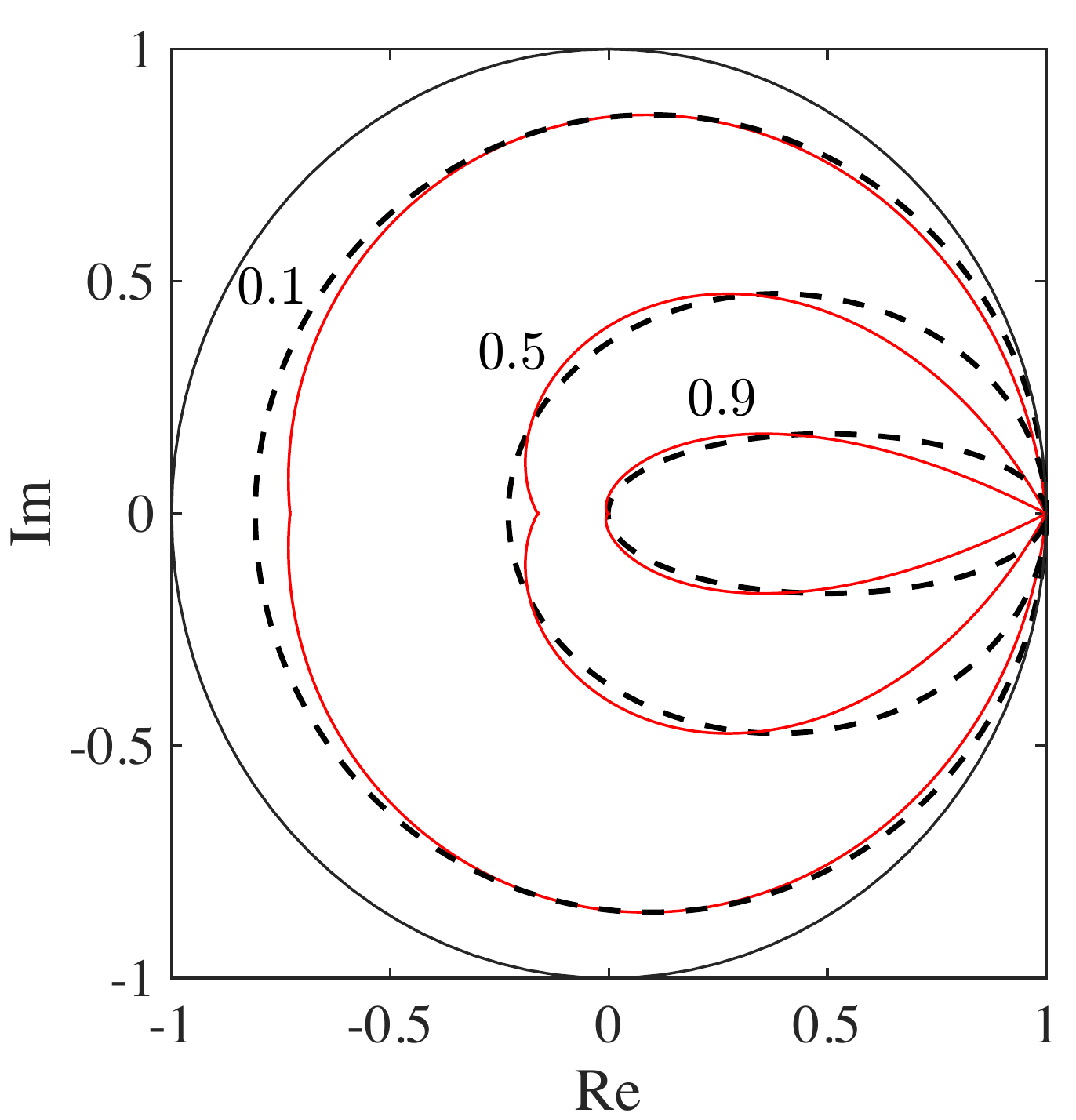} 
\end{center}
\caption{Convex approximations of the auxiliary information $O^{\textrm{max}}_{s}$ for the cardioids in red for $\zeta^{\rm{min}} =\{0.1, 0.5, 0.9\}$ using more conservative ellipses (Fact \ref{prop:ellipse_ricco}) in dashed black. For comparison, see Fig. \ref{fig:mp_circles_ellipses}.}
\label{fig:mp_region_z_ellipse3}                                 
\end{figure}  

\subsection{Period between damped oscillations} \label{sec:period_oscillations}
Initially, assume that $T_{s}\le T^{\rm{max}}_{d}/4$ such that the eigenvalues of $A$ are located on the right-half plane of the $z$-plane.
For instance, consider the cases $w^{1,1}_{d}$ and $w^{1,2}_{d}$ (where $T^{1}_{s} <T^{2}_{s}$) illustrated in Fig. \ref{fig:mapping_regions2}b. The region described by $T_{d}$ on the $z$-plane should be described 
by a conic section at the origin and with inner angle $\theta^{\rm{max}}=w^{\rm{max}}_{d}T_{s}$, where $-\pi/2\le \theta^{\rm{max}} \le \pi/2$. The following result is a straightforward extension from \cite{Chilali1996}.
 \begin{fac}\label{fact:conic}The set that describes 
 the region corresponding to $w_{d}^{\rm{max}}$ by means of a conic sector at the origin and with inner angle $\theta^{\rm{max}}$ is given by
\begin{eqnarray}\label{def_conic_sector}
\mathcal{D}_{T}=\left\{z~\in~\mathbb{C}: \rm{tan}(\theta^{\rm{max}}) \rm{Re}(z)\geq \left|\rm{Im}(z)\right| \right\}
\end{eqnarray}
and is equivalent to the LMI region $f_{\mathcal{D}_{T}}(z) \geq 0$,
\begin{eqnarray}\label{eq:T_1} \nonumber 
f_{\mathcal{D}_{T}}(z)=0_{2 \times 2}&+&\left[\begin{array}{cc}
              \rm{sin}(\theta^{\rm{max}}) & -\rm{cos}(\theta^{\rm{max}})  \\
                \rm{cos}(\theta^{\rm{max}})  & \rm{sin}(\theta^{\rm{max}}) 
              \end{array}
\right]z\\ 
&+&\left[\begin{array}{cc}
                \rm{sin}(\theta^{\rm{max}}) & \rm{cos}(\theta^{\rm{max}})  \\
                -\rm{cos}(\theta^{\rm{max}})  & \rm{sin}(\theta^{\rm{max}}) 
              \end{array}
\right]\bar{z}.
\end{eqnarray}
\end{fac}

The LMI region described by the Fact \ref{fact:conic} should be used to perfectly fit the region described by the period between damped oscillations  $T_d$ on the right-half part of the $z$-plane as presented in Fig. \ref{fig:mapping_regions2}b. However, for $w_d > w_d^{2,1}$, that is, $| \theta^{\rm{max}} | > \pi/2$ (see Fig. \ref{fig:mapping_regions2}b), the left-half part of the dynamical region cannot be written as a LMI region, because the corresponding region is not convex.

\begin{remark}
Recall that $\theta^{\rm{max}}=w^{\rm{max}}_{d}T_{s}$ and $\theta^{\rm{max}}\le\pi/2$ in the right-half part of the $z$-plane. From the latter we obtain $w^{\rm{max}}_{d}T_{s}\le\pi/2$. Replacing $w^{\rm{max}}_{d}=2\pi/T^{\rm{max}}_{d}$ in $w^{\rm{max}}_{d}T_{s}\le\pi/2$, we have that $T_{s}\le T^{\rm{max}}_{d}/4$,  meaning that, in our procedure, we need at least four samples by each damping period. From Nyquist's sampling theorem, each period must be sampled at least two times to avoid aliasing. However, in practice, the golden rule is to sample from six to ten times per period \cite{Gopal2003}. Thus, the non-convexity of left-half part is an issue only for poorly sampled systems, for which the LMI defined in Fact \ref{fact:conic} is not useful for subspace identification.
\end{remark}

\subsection{Settling time} \label{sec:settling_time} 
The region described by $t_{s}$ on the $z$-plane is given by a circle whose center and radius are given by  
\begin{eqnarray}\label{eq:c_s}
c_{s}&\triangleq &0,\\ \label{eq:rmax_s}
r_{s}^{\textrm{max}}&\triangleq& e^{-{\zeta w_{n}}^{\rm{min}}T_{s}}.
\end{eqnarray}

Therefore, the LMI defined in Fact \ref{fact:circle} can be used here replacing $c^{\rm{max}}$ by the new center $c_{s} = 0$ and replacing $r^{\rm{max}}$ by the radius $r_{s}^{\textrm{max}}$ in order to perfectly fit  the region shown in Fig. \ref{fig:mapping_regions3}b.  

Note that the Fact \ref{fact:circle} can be also used to ensure stability for the estimated model by setting $c_{s} = 0$ and $r_{s}^{\textrm{max}}=1$ in \eqref{eq:c_s}-\eqref{eq:rmax_s}. 
%
\subsection{Practical aspects to build LMI regions}\label{Sec:Bound_variables}
From sections \ref{background} and \ref{sec:Overshoot_revisited}-\ref{sec:settling_time}, we know that the auxiliary information related to $\hat{\zeta}$, $\hat{w}_{d}$ and $\widehat{\zeta w_{n}}$ is estimated using step response test data and the relations \eqref{eq:Mp_max}-\eqref{eq:ts_max}. Note that if more than one step response test is available, then the estimated auxiliary information can be obtained by means of the average of such parameters over the available tests. Conversely, we can also estimate $\hat{\zeta}$, $\hat{w}_{d}$ and $\widehat{\zeta w_{n}}$ from the average of the step response tests. In addition, tuning variables $\Delta_{\zeta}$, $\Delta_{w_{d}}$ and $\Delta_{\zeta w_{n}}$ for such parameters can be set by the standard deviation from the average value of the correspondent variables. Recall that during the identification process, the user should tune these parameters parsimoniously.

Since the region shown in Fig. \ref{fig:mapping_regions1}b is not convex, some convex approximations are proposed in \cite{botto2002,Rosinova2014} for control systems. In the LMI regions discussed in Section \ref{sec:convex_approx} for system identification, it is reasonable to be more conservative with the usage of uncertain \textit{prior} information. Recall that the \textit{prior} information may be uncertain due to noise and the fact that the corresponding dynamical regions presented in Section \ref{background} 
are exact only for second-order linear systems \eqref{eq:segunda_ordem}. To handle that, we set the parametrization of the LMI dynamical regions as follows:
\begin{eqnarray}\label{eq:zeta_max}
\zeta^{\rm{min}} &\triangleq& \hat{\zeta}-\Delta _{\zeta}, \\ \label{eq:wd_max}
w^{\rm{max}}_{d} & \triangleq& \hat{w}_{d}+\Delta_{w_{d}} ,\\ \label{eq:zeta_wn_max}
{\zeta w_{n}} ^{\rm{min}}& \triangleq& \widehat{\zeta w_{n}}-\Delta_{\zeta w_{n}} ,
\end{eqnarray}
where $0<\zeta^{\rm{min}} <1$, $0<w^{\rm{max}}_{d}<w_{s}/4$ and ${\zeta w_{n}} ^{\rm{min}}>0$, and  $\Delta _{\zeta}$, $\Delta_{w_{d}} $ and $\Delta_{\zeta w_{n}}$ are defined by the user as pointed out above. 
In Fig. \ref{fig:bound_regions} regions related to \eqref{eq:zeta_max}-\eqref{eq:wd_max} are exemplified. Observe that the effect of the tuning variable given by \eqref{eq:zeta_wn_max} is analogous to the tuning variable given by \eqref{eq:zeta_max}. In this case, note that the smaller the parameter is, the bigger is the area of the correspondingly parameter on the $z$-plane.

\begin{remark}
In fact, the choice of $\zeta^{\rm{min}}$, $w^{\rm{max}}_{d}$ and $\zeta w_{n}^{\rm{min}}$ is dependent of the process design and how deep is the knowledge about the \textit{prior} information of the process. On the other side, if the poles of the model are estimated in a more conservative region, then we obtain more degrees of freedom on the estimation procedure and also more chance to find poles near to the region of the dominant poles. For this reason, we choose the signals and the superscript max-min in the parameters of \eqref{eq:zeta_max}-\eqref{eq:zeta_wn_max}. So, it is crucial to observe that there is no guarantee that the estimated parameters $\hat{\zeta}$, $\hat{w}_{d}$ and $\widehat{\zeta w_{n}}$ are the true values. However, they are estimated parameters that can be a source of auxiliary information of the dominant dynamic of the system. In addition, the validation process is crucial on this step, deciding if the performance of the model is improved or not with the usage of auxiliary information.  
\end{remark}
In Procedure \ref{proc:6_1} we sum up all the steps and related equations in order to solve the problem of mapping constraints in the \textit{subspace identification with eigenvalue constraints} problem.

\begin{algori} {Constrained subspace identification: mapping constraint regions onto discrete-time}  \label{proc:6_1}
\\ 
\textbf{$1^{\rm{st}}$ Step} From dynamical data, estimate $\{A^{*},B^{*},C^{*},D^{*}\}$ by means of a standard subspace identification method.

\textbf{$2^{\rm{nd}}$ Step} Evaluate the step response in order to estimate the values of $O_{s}$, $T_{d}$ and $t_{s,\rm{1}\%}$.

\textbf{$3^{\rm{rd}}$ Step} Using \eqref{eq:Mp_max}-\eqref{eq:ts_max} obtain the parameters $\hat{\zeta}$, $\hat{w}_{d}$ and $\widehat{\zeta w_n}$, respectively, from the values estimated in the previous step.

\textbf{$4^{\rm{th}}$ Step} Building LMI regions: {\rm(i)} Estimate $\beta^{\rm{max}}=\rm{cos}^{-1}(\hat{\zeta})$ and set \eqref{eq:cmax}-\eqref{eq:rmax}. Then, determine \eqref{eq:a_major}-\eqref{eq:f_ellipse} in order to set \eqref{eq:mp_2} building the ellipse given by Fact \ref{fact:ellipse_rosinova}. Next, calculate ${\displaystyle \mu=\frac{1}{a+c^{\rm{max}}}}$, determine the new major axis as $a_{n}=a\mu$ and the new center as $c^{\rm{max}}_{n}=c^{\rm{max}}\mu$ in order do set \eqref{eq:mp_2} and build the new ellipse region given by Fact \ref{prop:ellipse_ricco}. {\rm(ii)} Considering the sampling period $T_{s}$, obtain \hbox{$\theta^{\rm{max}}=\hat{w}_{d}T_{s}$} and set \eqref{eq:T_1} in order to build the conic region defined in Fact \ref{fact:conic}. {\rm(iii)} Finally, set \eqref{eq:c_s}-\eqref{eq:rmax_s} and form the circle region defined in Fact \ref{fact:circle}.

\textbf{$5^{\rm{th}}$ Step} Based on prior information, evaluate the use of the dynamical regions related to $\hat{\zeta}$, $\hat{w}_{d}$ and $\widehat{\zeta w_n}$ and the tuning variables  $\Delta _{\zeta}$, $\Delta_{w_{d}} $ and $\Delta_{\zeta w_{n}}$ in order to estimate more conservative regions \eqref{eq:zeta_max}-\eqref{eq:zeta_wn_max}. Combine the LMI regions by means of Lemma \ref{lem:intersection} and form the constraint \eqref{eq:D_constraint}.

\textbf{$6^{\rm{th}}$ Step} Solve the problem of subspace identification with eigenvalue constraints \eqref{eq:eigenvalue_c_problem}-\eqref{eq:P_matrix} and obtain $\hat{A}$.

\textbf{$7^{\rm{th}}$ Step} Validate the constrained estimated model and evaluate the necessity to return to the $5^{th}$ Step.
\end{algori}
 \begin{figure}[H]
\begin{center}
\includegraphics[scale=0.5]{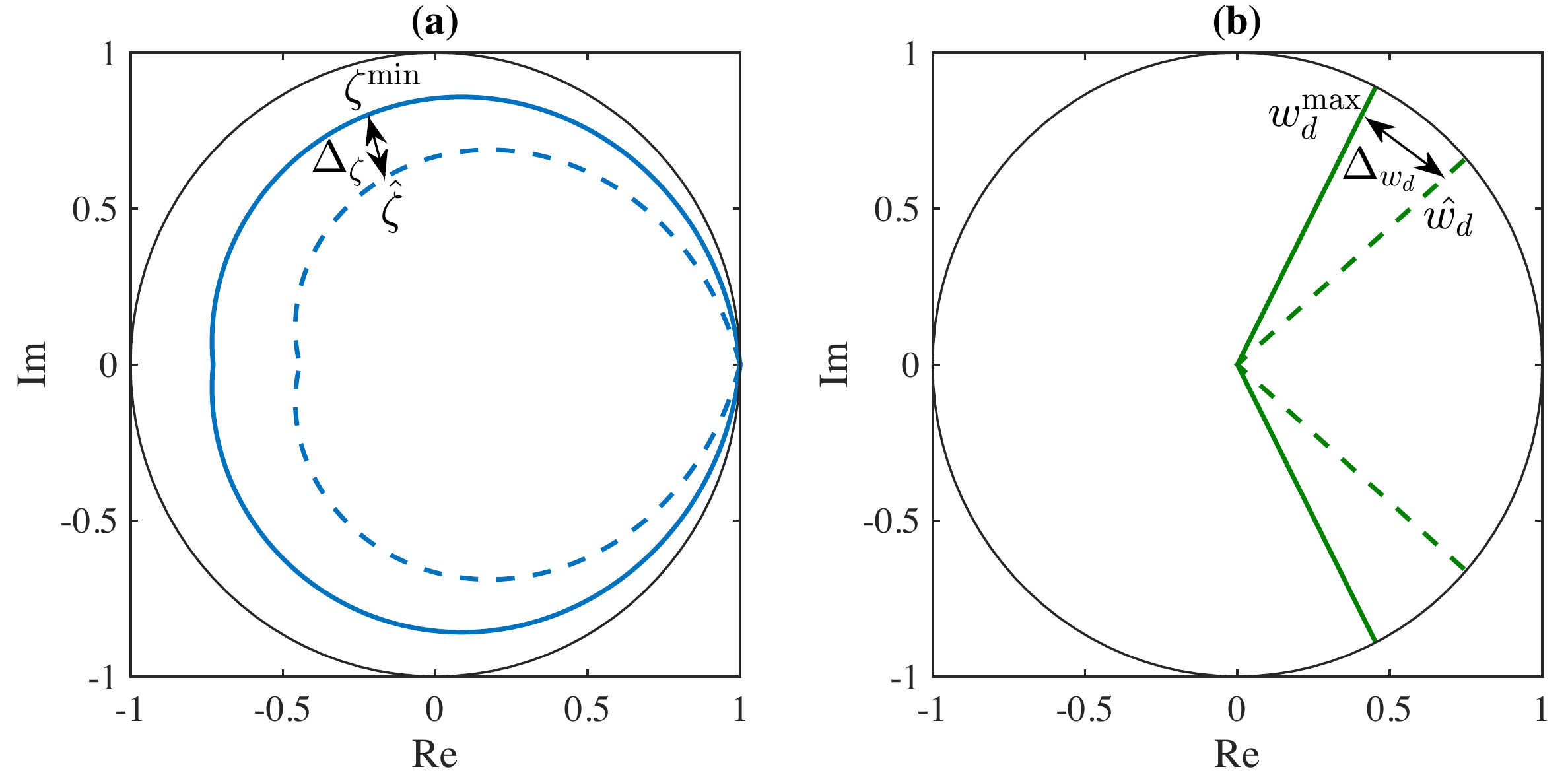} 
\end{center}
\caption{Examples of the effect of the max-min variables defined in \eqref{eq:zeta_max}-\eqref{eq:wd_max} on the $z$-plane.  (a) Overshoot: in dashed blue the area of the respective estimated variable $\hat{\zeta}$ and in continuous blue the new area $\zeta^{\rm{min}}$ defined with the tuning variable $\Delta _{\zeta}$.  (b) Period between damped oscillations: in dashed green the area of the respective estimated variable $\hat{w}_{d}$ and in continuous green the new area $w_{d}^{\rm{max}}$ defined with tuning variable $\Delta _{w_{d}}$. The main objective to create a border larger than the estimated one is constraint the poles in a more conservative region. The more conservative region related to the settling-time \eqref{eq:zeta_wn_max} can be analyzed similarly to the more conservative region defined based on the overshoot from \eqref{eq:zeta_max}.}
\label{fig:bound_regions}                                 
\end{figure}  
\section{Simulated results}\label{results}
Our first example indicates that using the Fact \ref{prop:ellipse_ricco} is more suitable to take into account information on how oscillatory a system is than the LMI inner approximations presented in \cite{Rosinova2014}. Indeed, \cite{Teixeira2011} argues that prior information may damage the model quality if its uncertainty is not properly accounted for. Likewise, our second example corroborates this result, indicating that one should be more conservative on the definition of the LMI eigenvalue regions.
Finally, we illustrate that even though the LMI eigenvalue regions here revisited or presented are approximations for higher-order systems, they may be useful to improve their corresponding models.

In the following examples, the unconstrained estimates are obtained by means of the PI-MOESP method proposed by \cite{verhaegen2007} with past and future horizon lengths set equal to 10. YALMIP \cite{lofberg2004} was used to solve the convex constrained optimization problems with MOSEK \cite{mosek} as the selected solver, both packages running in 
MATLAB.
\subsection{Building LMI regions for subspace identification}\label{subsec:LMI_r}
Consider the second-order continuous-time linear system \eqref{eq:segunda_ordem} with $K=0.7$, $\zeta=0.2$ and $w_{n}=1$. The eigenvalues of the corresponding state matrix $A$ from \eqref{state_eq} are $s_{1,2 }=-0.2 \pm  0.9798\jmath$. The output is measured with the sampling period $T_{s}=0.3$ s and contaminated with colored noise $\nu$ generated by white noise $v$ with standard deviation $\sigma_{v}=1$ filtered by 

\begin{eqnarray}
\nu(s)=\frac{10s^{2}+5}{s^3+10s^{2}+s+2}v(s).
\end{eqnarray}

In order to identify the system, we generated a PRBS signal with $16$ bits and with values held during $100$ samples. The simulation is taken during $40$ seconds (not shown for brevity).   
We investigate a $100$-run Monte Carlo simulation with different noise realizations for $\sigma_{v}=1$. 
The estimated eigenvalues of the unconstrained estimation are shown in blue in Fig. \ref{fig:poles_ex1}. Note that the unconstrained estimator (in blue) fails by yielding unstable models at times. 
\begin{figure}[H]
\begin{center}
\includegraphics[scale=0.675]{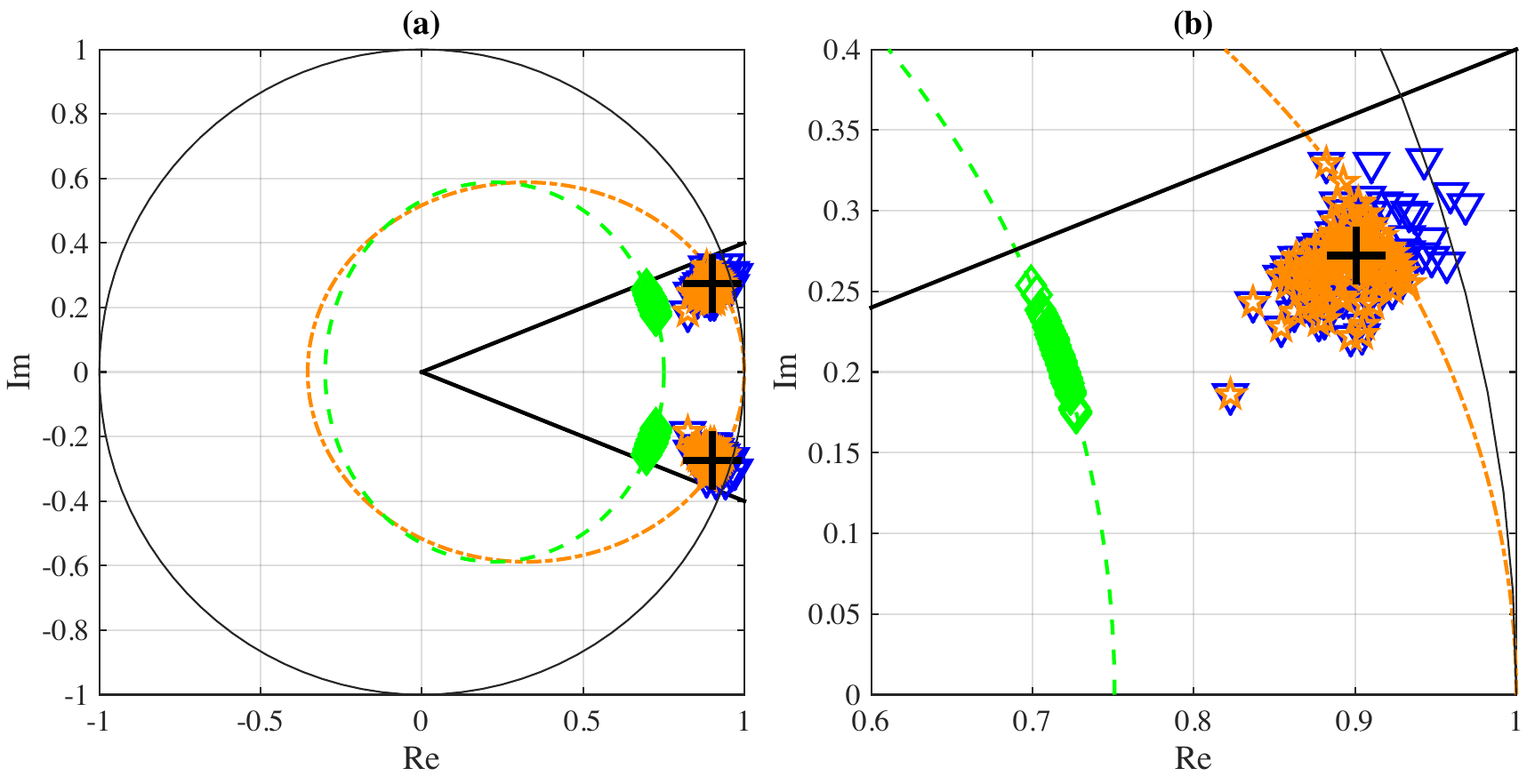} 
\end{center}
\caption{(a) Dynamical regions and estimated eigenvalues on the $z$-plane and (b) a zoom. The unconstrained estimates from PI-MOESP are shown by downward-pointing triangle ($\triangledown$) markers in blue. The constrained estimates using the Fact \ref{fact:ellipse_rosinova}  (dashed green ellipse) and Fact \ref{fact:conic} (full black conic lines) are shown by diamond markers ($\diamond$) in green. The constrained estimates using the Facts \ref{fact:conic} (full black conic lines) and \ref{prop:ellipse_ricco} (dash-dot line orange ellipse) are indicated by pentagram markers ($\star$) in orange. The true eigenvalues are given by plus markers (+) in black.
}
\label{fig:poles_ex1}                                 
\end{figure}
\begin{figure}[H]
\begin{center}
\includegraphics[scale=0.6]{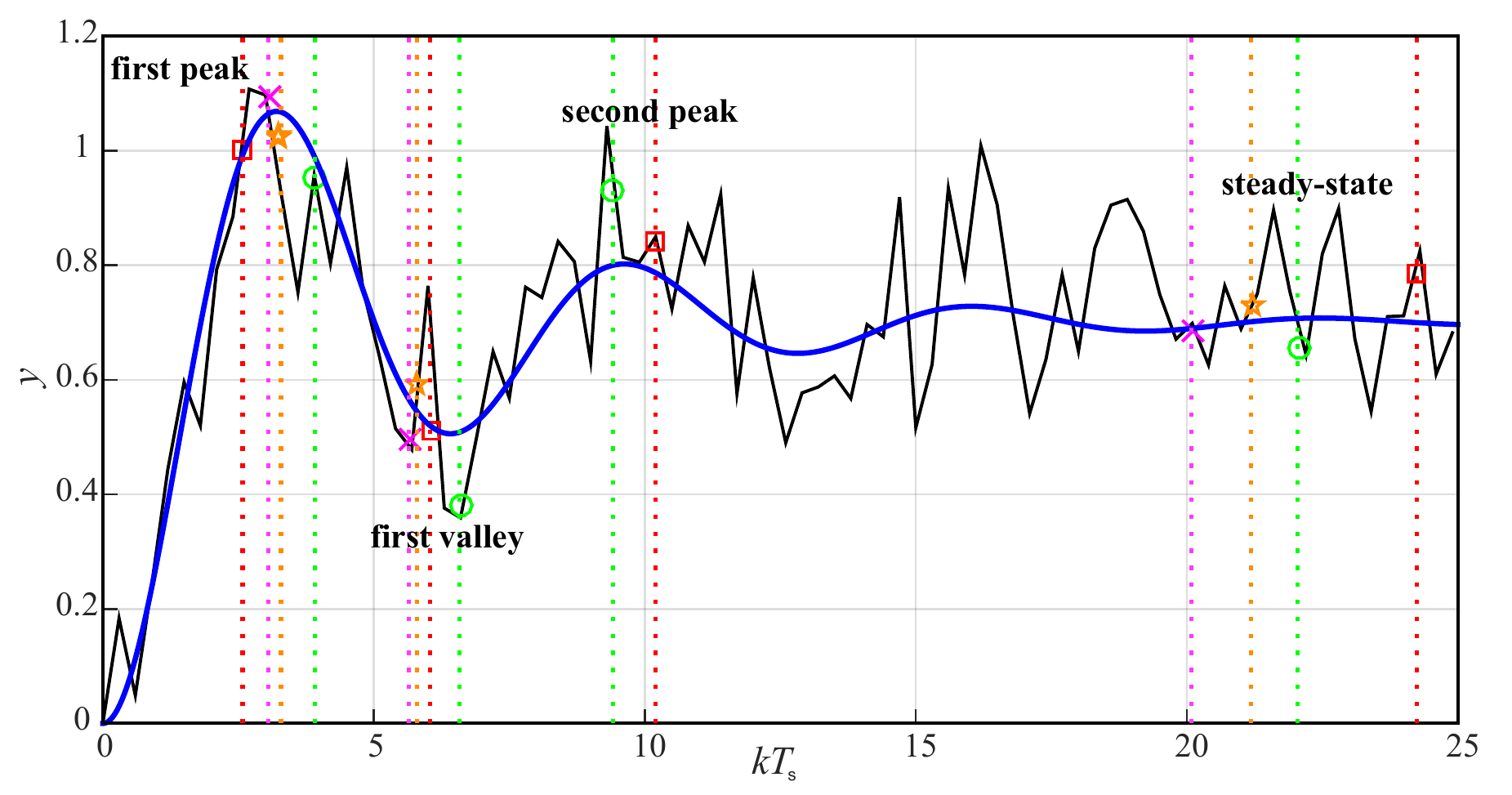} 
\end{center}
\caption{Output sequence $y_{k}$ ($\rm{SNR} \approx 5 \rm{dB}$) yielded by the step response simulation of the sampled model \eqref{eq:segunda_ordem} ($K=0.7$, $\zeta=0.2$ and $w_{n}=1$). For comparison, the continuous-time noise-free response is shown in blue. Markers selected in order to estimate the overshoot and the period between damped oscillations: the orange markers and magenta markers show the first peak, first valley and steady-state; the green and the red markers show the first peak, first valley, second peak and steady-state values. Each set of points, grouped by colors, are used to build a LMI region.
}
\label{fig:ex_2_a_step_e}                                 
\end{figure}
Next, we extract \textit{prior} information about the overshoot (related to $\zeta$) and the period between damped oscillations ($w_{d}$) from the step response tests and use them as constraints. Observe that, in this first example, for simplicity we set $\Delta_{\zeta}=0$, $\Delta_{\zeta w_n}=0$ and $\zeta w_n = 0$. In doing so, we follow the Procedure \ref{proc:6_1}.

Fig. \ref{fig:ex_2_a_step_e} shows both the ideal and noisy sampled step response for $\rm{SNR} \approx 5 \rm{dB}$. 
Observe that the points that are critical to the estimation of the overshoot and the period between damped oscillations are indicated by orange star markers. 
We use the half period $T_{d}/2$ (the first peak and the first valley) to estimate the period between damped oscillations. Also, we use the first peak and the steady-state to estimate the overshoot. Based on these points, we obtain the parameters $\hat{\zeta} \approx 0.36$ and $\hat{w}_{d}\approx 1.27$ rad/s. We use such \textit{prior} information for the direct parameterization of the LMI dynamical regions given by \eqref{eq:mp_2} in Fact \ref{fact:ellipse_rosinova} and in Fact \ref{prop:ellipse_ricco}, and \eqref{eq:T_1} in Fact \ref{fact:conic}.
The constrained estimates are obtained using the LMI regions defined by the Facts \ref{fact:ellipse_rosinova}, \ref{fact:conic} and  \ref{prop:ellipse_ricco}. 

These LMI regions and the corresponding estimated eigenvalues of the unconstrained and constrained models are shown in Fig. \ref{fig:poles_ex1}. 
Observe that, unlike our proposed LMI region (Fact \ref{prop:ellipse_ricco}), the inner region (in green) given by the Fact \ref{fact:ellipse_rosinova} \cite{Rosinova2014} does not encompass an important eigenvalue region nearby the unit circle. Fig. \ref{fig:bode_moesp_pi_test} shows the frequency response of the estimated models. The results suggest that the proposed Fact \ref{prop:ellipse_ricco} is useful in subspace identification problems.
\begin{figure}[H]
\begin{center}
\includegraphics[scale=0.65]{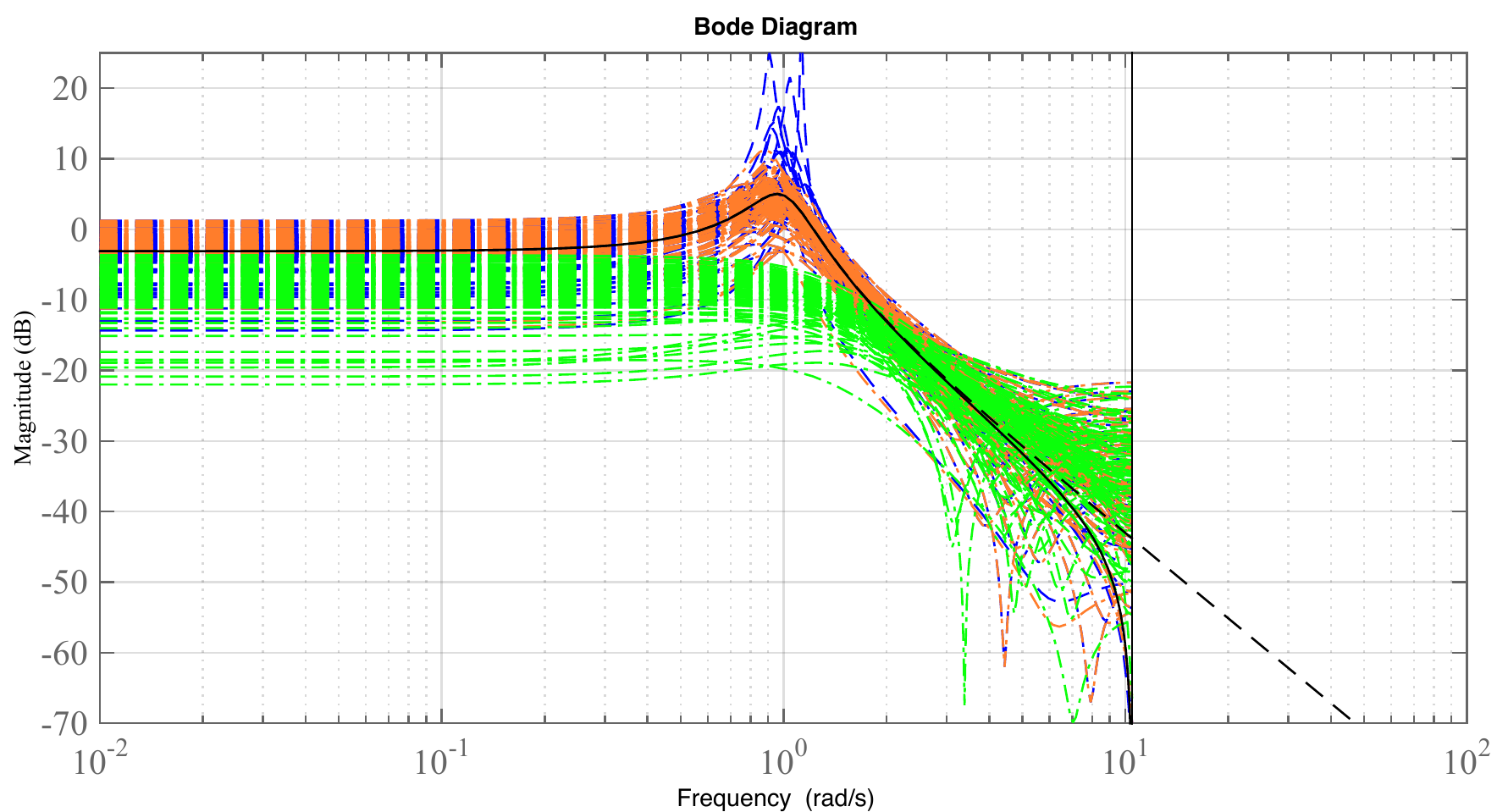} 
\end{center}
\caption{Comparison of frequency response estimates: true response (black line); using unconstrained PI-MOESP method (dashed blue); using Fact \ref{fact:ellipse_rosinova} (dash-dot green)  and Fact \ref{fact:conic}; and using the Facts \ref{fact:conic} and \ref{prop:ellipse_ricco} (dash-dot orange.)}
\label{fig:bode_moesp_pi_test}                                 
\end{figure}
\subsection{The effect of setting inaccurate LMI regions}
Consider again the system simulated in Section \ref{subsec:LMI_r}. 
Now we aim at investigating the effect of setting an inaccurate LMI region. 
For example, to compare the influence of the error on the estimation of the overshoot and the period between damped oscillations, we consider three cases as shown in Fig. \ref{fig:ex_2_a_step_e} by red, green and magenta markers. Again note that we set $\Delta_{\zeta}=0$, $\Delta_{\zeta w_n}=0$ and $\zeta w_n = 0$ and that we follow the Procedure \ref{proc:6_1} in this example.

 For the first and second cases (green and the red markers), we use the first valley and second peak for the estimation of the parameters $\hat{w}_{d,1}\approx  1.11$  rad/s and $\hat{w}_{d,2}\approx 0.77$  rad/s. In the third case (magenta markers) we consider the first peak and first valley for the estimation of the parameter $\hat{w}_{d,3} \approx 1.21$ rad/s.  For both cases, we consider the first peak and the steady-state value on the estimation of the parameters $\hat{\zeta}_{1}\approx 0.33$ (green),  $\hat{\zeta}_{2}\approx 0.43$ (red) and $\hat{\zeta}_{3} \approx 0.24$ (magenta). Such parameters can be used as \textit{prior} information to build the LMI dynamical regions given by \eqref{eq:mp_2} (Fact \ref{prop:ellipse_ricco}) and \eqref{eq:T_1} (Fact \ref{fact:conic}). Fig. \ref{fig:figure_2_ex_2} shows the corresponding regions in the same color of the aforementioned markers.
 \begin{figure}[H]
\begin{center}
\includegraphics[scale=0.685]{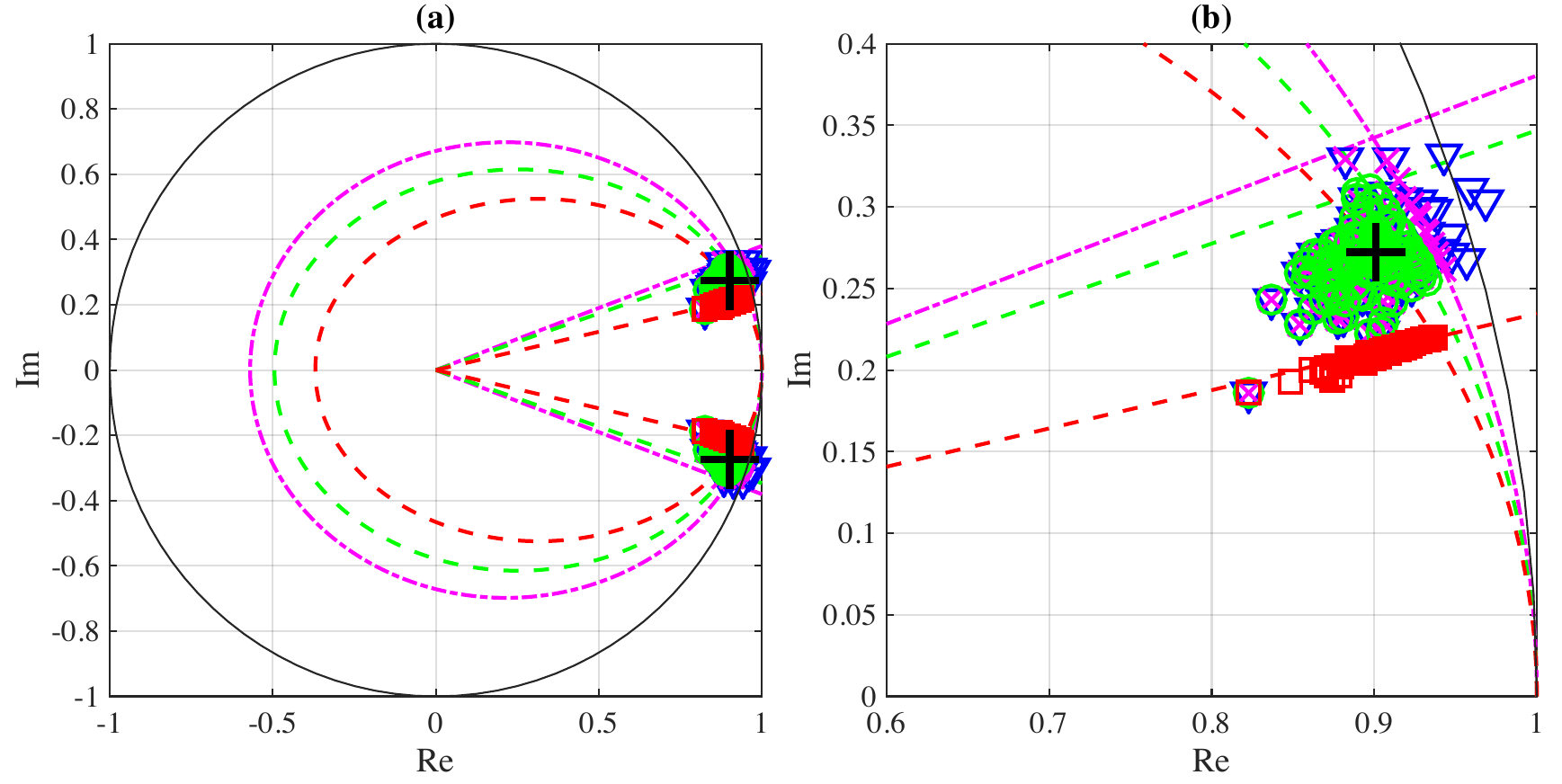} 
\end{center}
\caption{(a) Dynamical regions and estimated eigenvalues on the $z$-plane and (b) a zoom. The unconstrained estimates  are shown by downward-pointing triangle ($\triangledown$) markers in blue.  The constrained estimates using the intersection of the Fact \ref{fact:conic} (cones) and the Fact \ref{prop:ellipse_ricco} (ellipses) are shown for three cases by means of square markers ($\square$) in red, by circle markers in green and cross markers magenta. The true eigenvalues are given by plus markers (+) in black. }
\label{fig:figure_2_ex_2}                                 
\end{figure}

We perform Monte Carlo simulation comprising $100$ noise realizations. The results shown in Fig. \ref{fig:figure_2_ex_2}  indicates that using the constraints \eqref{eq:mp_2} (Fact \ref{prop:ellipse_ricco}) and \eqref{eq:T_1} (Fact \ref{fact:conic}) yields improved results for the first (green) and third (magenta) cases compared to the unconstrained case (blue). In fact, there are no significant differences between the latter cases nearby the true eigenvalues of the system; see Fig. \ref{fig:figure_2_ex_2}b. The second case (red) shows the effect of 
setting a smaller LMI region related to the period between damped oscillations. 
Therefore, the use of tuning variables proposed in \eqref{eq:zeta_max}-\eqref{eq:zeta_wn_max} is a more conservative choice, and thus more appropriate, for the tuning of the LMI regions as we shall see next.

\subsection{Setting conservative LMI-regions for higher-order dynamics}\label{ex:conservative_regions}
Consider now the system given by
\begin{eqnarray}\label{eq:fourth_order_model}
G(s)=\frac{100 s + 1500}{s^4 + 11 s^3 + 130 s^2 + 1020 s + 2000},
\end{eqnarray}
whose poles are $s_{1,2}=-0.50 \pm 9.99\jmath$, \hbox{$s_{3}=-7.24$} and $s_{4}= -2.76$.
The output is  measured with  $T_{s}=0.05 $ s, corrupted by the colored noise $\nu$, which is generated by filtering the white noise $v$ with standard deviation $\sigma_{v}=1$  as 
\begin{eqnarray}
\nu(s)=\frac{10 s^2 + 2}{s^3 + 10 s^2 + s + 20}v(s).  
  \end{eqnarray}

The PRBS input signal is generated with $20$ bits and hold for $100$ samples. The simulation is performed along $20$ seconds yielding $40000$ samples (not shown). From one of the steps of the PI-MOESP, we set the order of the model as four. We perform a $100$-run Monte Carlo simulation. Fig. \ref{fig:figure_ex3} shows the unconstrained estimates in blue. Note that the unconstrained (blue) estimator may fail on the eigenvalue localization. 

\begin{figure}[H]
\begin{center}
\includegraphics[scale=0.42]{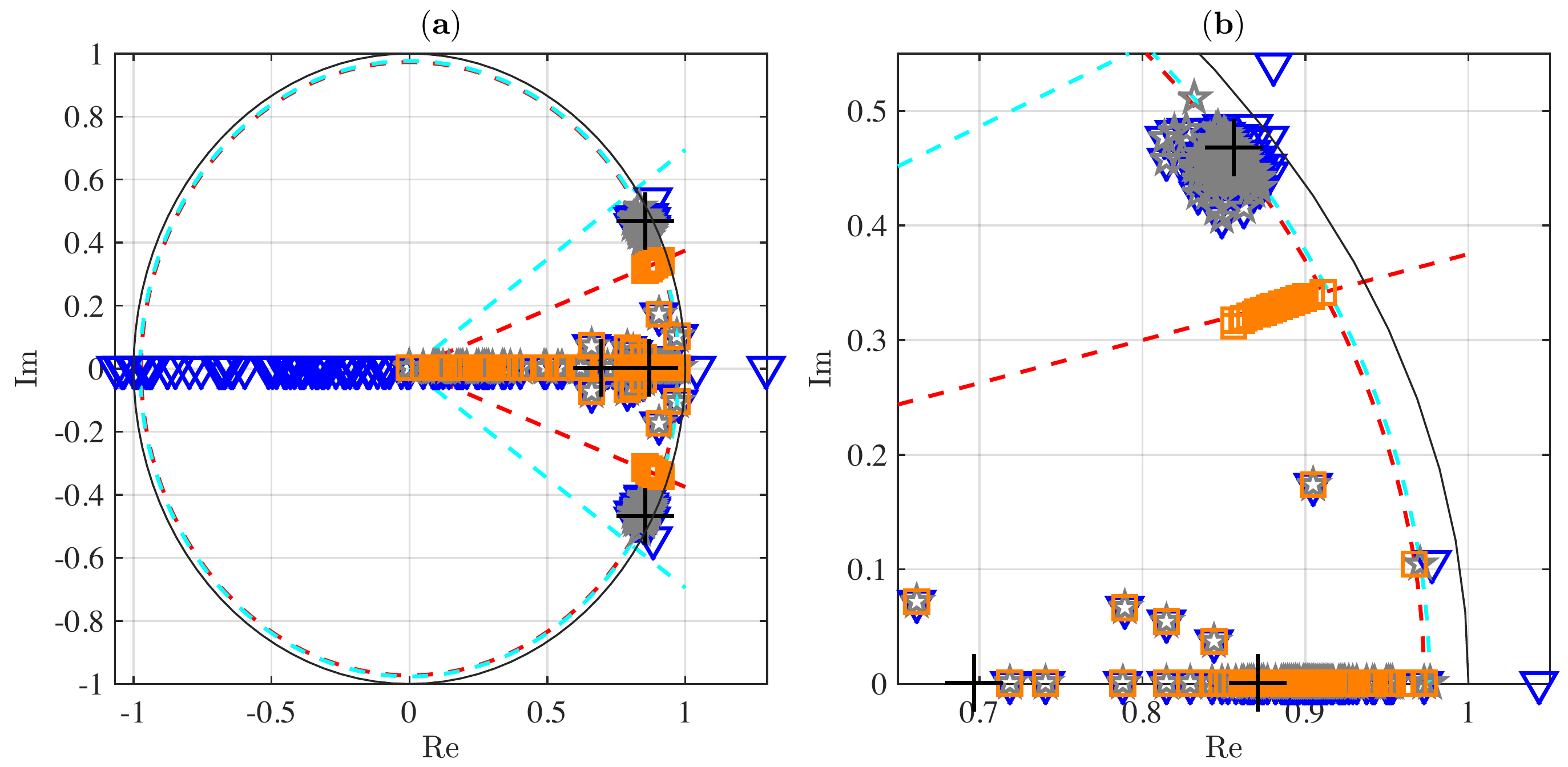} 
\end{center}
\caption{(a) Dynamical regions and estimated eigenvalues on the $z$-plane and (b) a zoom. The unconstrained estimates are shown by triangle ($\triangledown$) markers in blue. The constrained estimates using both Fact \ref{fact:conic} (dashed red conic lines) and  Fact \ref{fact:circle} for \eqref{eq:c_s}-\eqref{eq:rmax_s} (dashed red circle) are shown by square markers ($\square$) in orange. The constrained estimates using the tuning variables \eqref{eq:zeta_max}-\eqref{eq:zeta_wn_max} are shown by dashed cyan conic lines, dashed cyan circle and by means of pentagram markers ($\star$) in gray. The true eigenvalues are shown by plus markers (+) in black.}
\label{fig:figure_ex3}                                 
\end{figure}

Thus, we use the auxiliary information obtained from a step-response experiment. Such experiment is used to estimate the period between damped oscillations and the settling-time. Note that the step response  (Fig. \ref{fig:ex_3_a_step_2}) of this system is more complex than the pattern for second-order systems. Considering that we are dealing with a higher order dynamics, from Fig. \ref{fig:ex_3_a_step_2}, it seems that there is a superposition of an underdamped response and an overdamped  response. Even so, the prior information of the settling-time and the period between oscillations can be approximated. Since we observe a small overshoot (for instance, see the region of $\zeta^{\rm{min}} =0.9$ in Fig. \ref{fig:mp_circles_ellipses}), we prefer not to use this information. Also it is important to note that the dominant dynamics of the step response test is not highly affected by the zero. Despite the fact that the effect of the zero in the dominant dynamics of the step response is not evident, we highlight that it may be possible in some cases. However, we try to approximate the regions related to the estimated auxiliary information that can be useful on the identification process as auxiliary information. If the approximated region does not contribute to the improvement of the perfomance of the constrained estimated model, then the auxiliary information can be discarded or tuned. In such cases, the validation step can contribute on this decision. 
The Procedure \ref{proc:6_1} summarizes all the steps followed to estimate the auxiliary information used to generate the dynamical LMI regions shown in Fig. \ref{fig:figure_ex3}. The step response of \eqref{eq:fourth_order_model} is shown in Fig. \ref{fig:ex_3_a_step_2} for both ideal and noisy cases ($\rm{SNR} \approx 5 \rm{dB}$). We highlight the important points for the estimation of the period between damped oscillations and settling-time by red markers. Here, we use the second peak and the second valley (half period $T_{d}/2$) for the estimation of the parameter $w^{1}_{d}$. First, we estimate the parameters $\hat{w}^{1}_{d}\approx 7.17$  rad/s and $\hat{t}^{1}_{s,\rm{1}\%}\approx  8.33$ s (which is equivalent to $\widehat{\zeta w^{1}_{n}} \approx 0.55$ rad/s) and build an LMI region using the Facts \ref{fact:conic} and \ref{fact:circle} according to \eqref{eq:c_s}-\eqref{eq:rmax_s}; see the dashed red region in Fig. \ref{fig:figure_ex3}a.

\begin{figure}[h]
\begin{center}
\includegraphics[scale=0.43]{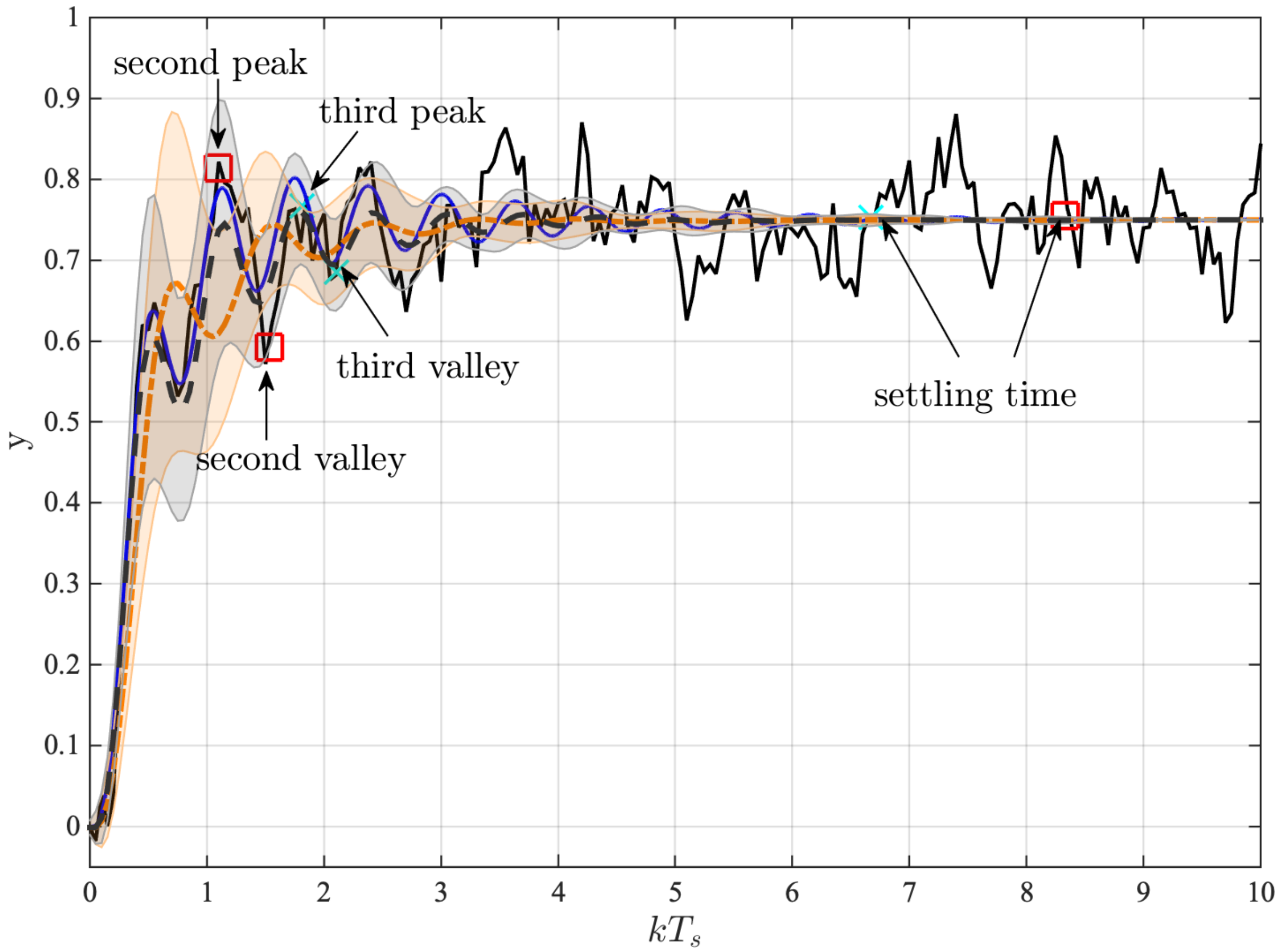} 
\end{center}
\caption{Output sequence $y_{k}$ ($\rm{SNR} \approx 5 \rm{dB}$) yielded by the step response simulation of the true model \eqref{eq:fourth_order_model} is shown in continuous black line. For comparison the continuous-time noise-free response is shown in continuous blue line. The red markers show the second peak, second valley and the settling time values selected to estimate $w^{1}_{d}$ and $t^{1}_{s,\rm{1}\%}$. The cyan markers show the third peak, third valley and the settling time values selected to estimate $w^{2}_{d}$ and $t^{2}_{s,\rm{1}\%}$. The average of 100-run normalized step responses obtained from models constrained in the red region of Fig. \ref{fig:figure_ex3} is shown in orange dashed-dot line and from models constrained in the cyan region of Fig. \ref{fig:figure_ex3} is shown in gray dashed line. The respective standard deviations are shown by the shaded regions in orange and gray.}
\label{fig:ex_3_a_step_2}                                 
\end{figure}

Next, in order to estimate a new conservative constrained region for the eigenvalues, in Fig. \ref{fig:ex_3_a_step_2} we choose the third peak, third valley and the settling time highlighted in cyan. Then, we estimate $\hat{w}^{2}_{d}\approx 10.49$ and $\hat{t}^{2}_{s,\rm{1}\%}\approx 6.67$ s (equivalent to $\widehat{\zeta w^{2}_{n}} \approx 0.69$ rad/s). These points are used to set a new region, for example, we define the following parametrization: ($\rm{i}$) $\hat{w}_{d}\triangleq(\hat{w}^{1}_{d}+\hat{w}^{2}_{d})/2$ and $\Delta_{w_{d}}\triangleq |\hat{w}^{1}_{d}-\hat{w}^{2}_{d}|$ in \eqref{eq:wd_max} and; ($\rm{ii}$) $\widehat{\zeta w_{n}} \triangleq(\widehat{\zeta w^{1}_{n}} +\widehat{\zeta w^{2}_{n}})/2$ and $\Delta_{\zeta w_{n}} \triangleq |\widehat{\zeta w^{1}_{n}} -\widehat{\zeta w^{2}_{n}} |$ in \eqref{eq:zeta_wn_max}.  These relations yield $\hat{w}^{\rm{max}}_{d}\approx 12.14$ rad/s and $\widehat{\zeta w_{n}}^{\rm{min}}\approx 0.48$ rad/s. Thus, the new constrained region is built by means of \eqref{eq:wd_max} and \eqref{eq:zeta_wn_max} and the Facts \ref{fact:conic} and \ref{fact:circle}  according to \eqref{eq:c_s}-\eqref{eq:rmax_s}; see the dashed cyan region in Fig. \ref{fig:figure_ex3}a. 

 
As indicated by Fig. \ref{fig:figure_ex3}, the use of \textit{prior} information is useful for subspace identification. Comparing the unconstrained PI-MOESP estimation (blue markers) and the constrained estimates (orange and cyan), we verify that the estimated complex and real eigenvalues are close to the true eigenvalues of \eqref{eq:fourth_order_model}. Also, we verify that the use of more conservative regions should be an effective way to circumvent some issues related to the inaccurate estimation of the approximated dynamical regions even for higher-order systems. The average of 100-run normalized step responses for both constrained estimations  with the respective standard deviation ($\sigma$) are shown in Fig. \ref{fig:ex_3_a_step_2}. As we expect, based on the estimated eigenvalues obtained in the regions (orange and cyan), qualitatively the relaxed estimations given in dashed gray has a better perfomance compared to the other in dashed-dot orange.

\section{Conclusions} \label{Conclusions}
This paper addresses the problems of ($\rm{i}$) mapping auxiliary information obtained from step-response experiments or historical input-output data onto useful LMI conservative regions for discrete-time state-space models and ($\rm{ii}$) using this information in subspace identification methods with eigenvalue constraints. 

We discuss the meaning of the auxiliary information about overshoot, the period between oscillations and the settling time in both $s$ and $z$ complex regions. In this regard, we argue that it is  simpler to extract prior information from the step-response experiments than by means of first principles for instance. In fact, even though the mapping of these properties is obtained for second order systems,  we can also use these approximated regions to constraint the eigenvalues of  more complex systems. If the auxiliary information is properly addressed, we recommend the use of the constrained method since  it guarantees at least the stability of the model. The numerical examples here discussed corroborates the aforementioned insights. 

A drawback in the use of LMIs on the constrained discrete-time subspace identification is the impossibility of constraining each eigenvalue separately. This issue is a topic of our future research efforts.

\section{Acknowledgments}
This research was supported by the Brazilian agencies:  National Council for Scientific and Technological Development (CNPq) and  Coordination for the Improvement of Higher Education Personnel (CAPES).
\bibliographystyle{iet}

\end{document}